\documentclass[lettersize,journal]{IEEEtran}
\usepackage{amsmath,amsfonts}
\usepackage{array}
\usepackage[caption=false,font=normalsize,labelfont=sf,textfont=sf]{subfig}
\usepackage{textcomp}
\usepackage{stfloats}
\usepackage{url}
\usepackage{verbatim}
\usepackage{graphicx}
\usepackage{cite}
\usepackage{color}
\usepackage{mathtools, cuted}
\usepackage{amssymb}
\usepackage{amsthm}
\usepackage{balance}
\usepackage{mathrsfs}
\usepackage{comment}
\usepackage{enumitem}
\usepackage{mathdots}
\usepackage{xfrac}
\usepackage{bm}
\usepackage{array}
\usepackage{algorithm}
\usepackage{algorithmicx}
\usepackage{algpseudocode}
\usepackage{physics}
\usepackage{lipsum}
\usepackage{float}
\mathtoolsset{showonlyrefs}

\newtheorem{theorem}{Theorem}

\newtheorem{lemma}[theorem]{Lemma}

\newcommand{\prt}[1]{\ensuremath{\left(  #1 \right)  }}

\DeclareMathAlphabet\mathbfcal{OMS}{cmsy}{b}{n}

\usepackage{etoolbox}

\begin{document}

\title{The Quadrature Gaussian Sum Filter and Smoother for Wiener Systems}

\author{Angel L. Cede\~no, Rodrigo A. Gonz\'alez and Juan C. Ag\"uero
\thanks{This work was supported in part by the grants ANID-Fondecyt 3240181 and 1211630, and the ANID-Basal Project AFB240002 (AC3E).
}
\thanks{A. L. Cede\~no is with the Electrical Engineering Department of Universidad de Santiago de Chile, Santiago, Chile, and with the Advanced Center for Electrical and Electronic Engineering AC3E. R. A. Gonz\'alez is with the Department of Mechanical Engineering of Eindhoven University of Technology, Eindhoven, The Netherlands. J. C. Ag\"uero is with the Electronic Engineering Department of Universidad T\'ecnica Federico Santa Mar\'ia, Valpara\'iso, Chile, and with the Advanced Center for Electrical and Electronic Engineering AC3E. Emails: angel.cedeno@sansano.usm.cl, r.a.gonzalez@tue.nl, juan.aguero@usm.cl.}}

\maketitle

\begin{abstract}
Block-Oriented Nonlinear (BONL) models, particularly Wiener models, are widely used for their computational efficiency and practicality in modeling nonlinear behaviors in physical systems. Filtering and smoothing methods for Wiener systems, such as particle filters and Kalman-based techniques, often struggle with computational feasibility or accuracy. This work addresses these challenges by introducing a novel Gaussian Sum Filter for Wiener system state estimation that is built on a Gauss-Legendre quadrature approximation of the likelihood function associated with the output signal. In addition to filtering, a two-filter smoothing strategy is proposed, enabling accurate computation of smoothed state distributions at single and consecutive time instants. Numerical examples demonstrate the superiority of the proposed method in balancing accuracy and computational efficiency compared to traditional approaches, highlighting its benefits in control, state estimation and system identification, for Wiener systems.
\end{abstract}


\section{Introduction}
\IEEEPARstart{T}{he} dynamics of most physical systems inherently exhibit nonlinear behaviors, which are described by mathematical models derived from first principles \cite{cessenat2018} or using data-driven algorithms \cite{ljung1999}. Due to the high complexity of these models, simpler versions in the form of best linear approximations \cite{pintelon2012system} are often used for control, supervision, and identification tasks. More advanced models integrate linear time-invariant (LTI) dynamic subsystems with nonlinear static subsystems, known as Block-Oriented Nonlinear (BONL) models \cite{Billings2013,Schoukens2017}. Among these models, Wiener models, consisting of a linear dynamic block followed by a nonlinear static block, offer advantages in computational efficiency and ease of implementation in control and system identification~\cite{Zhu2002}.

Filtering and smoothing methods for dynamical systems find applications in various domains, ranging from power systems to cybersecurity and chemical processes \cite{Ji2021, Xiao2022, Valipour2022}, and they play a crucial role in calculating posterior distributions of system states based on noisy measurements \cite{anderson1979}. Promising areas of application also include battery state-of-charge estimation \cite{wang2024improved, alcantara2024li}, including its hardware implementation \cite{Nuculaj2024}, photovoltaic power forecasting \cite{Thaker2024}, and AI-based medical applications \cite{Zhong2024}. Particularly in the field of control and system identification, filtering and smoothing algorithms for Wiener systems are of special interest to estimate states for Model Predictive Control, or design parameter estimation algorithms such as the Maximum Likelihood method \cite{Wills2013, Cedeno2024Id, Finke2017}. Although the Kalman filter and the Rauch–Tung–Striebel smoother provide optimal solutions for linear and Gaussian systems \cite{kalman1960, rauch1965}, for Wiener systems, the challenge lies in obtaining closed-form estimators due to the computational intractability of certain integrals \cite{sarkka2013bayesian}. This challenge motivates the development of suboptimal solutions for nonlinear and non-Gaussian systems.

The most commonly used approach for filtering and smoothing system states is the Monte Carlo approach, consisting of Sequential Importance Sampling (SIS) methods, also called particle filters \cite{doucet2000sequential, gordon1993novel}. In particle filtering, a set of samples and weights is propagated through the nonlinear functions of the system, allowing for direct state estimation without explicitly knowing the posterior distribution. The advantage of these methods lies in their relative ease of implementation; however, they are computationally costly when a large number of samples is used and / or when the order of the system increases \cite{Kitagawa1996b,Elfring2021}. In addition to particle filtering, there are other filtering approaches for nonlinear state-space systems that can provide filtered PDF estimates under certain considerations. One effective technique is the Extended Kalman Filter \cite{Gelb1974}, where a Taylor approximation of the nonlinear system around an estimate of the state is used, followed by applying the standard Kalman filter recursions. A key advantage of this method is its relatively simple implementation and use. It provides accurate estimates when nonlinear functions exhibit smooth behavior and its computational cost is comparable to that of the standard Kalman filter. However, a notable drawback is that the Extended Kalman Filter requires nonlinear functions to be differentiable, as the Jacobian matrix must be computed for its operation. Another technique based on the Kalman filter is the Quadrature Kalman Filter \cite{arasaratnam2007discrete, Closas2012}, which employs points from the Gauss quadrature to linearize the nonlinear state and output functions through a linear regression, thus obtaining closed-form solutions of the filtered PDFs. One advantage of the Quadrature Kalman Filter is that, being 3-point based, it has a reasonably low computational cost for low system's orders and is easy to implement. However, its accuracy drops when the nonlinearities are not strictly monotonic, as in the saturation or dead zone cases. On the other hand, the Unscented Kalman Filter \cite{Julier1997, Wan2000} aims to directly approximate the mean and variance of the state distribution by considering a fixed number of sigma points that propagate through the nonlinear functions. Such filter provides better accuracy for highly nonlinear systems than the Extended Kalman Filter by avoiding explicit linearization, enhancing state estimation. However, it incurs a higher computational cost due to the use of multiple sigma points, which makes it less suitable for real-time applications with limited processing power. Other frameworks combine elements of the aforementioned approaches. For example, the marginalized Kalman filter uses both the standard Kalman and particle filters \cite{Vitetta2019}.

In this paper, we present a novel state filtering and smoothing approach for Wiener systems, building upon our previous work on linear systems with quantized output measurements \cite{Cedeno2021b,cedeno2021,Cedeno2023} to handle a wide class of static nonlinearities. The main contributions of this paper are:

\begin{enumerate}[label=C\arabic*]
\label{contribution1}
\item 
A closed-form recursive algorithm is proposed for state filtering in Wiener systems. The method models the conditional probability function of the measurements as a sum of weighted Gaussian distributions using the Gauss-Legendre quadrature rule. Such an approach entails both constant and strictly monotone subsets, covering any combination of typical nonlinearities such as saturation, deadzone, linear rectifier, and quantization.
\label{contribution2}
    \item Building on the two-filter approach \cite{kitagawa1994two}, a recursive algorithm is introduced for computing smoothing state distributions as Gaussian mixture models. Furthermore, a method is proposed to calculate the joint distributions of the state in two consecutive time instants, which is particularly useful for system identification \cite{gonzalez2023algorithm}. 
  \label{contribution3}
    \item Extensive simulation examples validate the proposed filtering and smoothing algorithms against state-of-the-art particle filter and Kalman-based methods. The results show that the proposed approach significantly improves the state estimation accuracy for Wiener systems, while achieving low computational costs.
    \item The proposed filter and smoother algorithms for Wiener systems are shown to be implementable with multiple parallel instances of Kalman filter and smoother algorithms. This approach enables efficient hardware implementation, significantly reducing computation time and making them suitable for embedded system applications.    
\end{enumerate}
The remainder of this paper is structured as follows. In Section \ref{sec:setup}, we introduce the problem setup. Section \ref{sec:gmm} provides a general form for the conditional PDF $p(y_t|\mathbf{x}_t)$ in terms of a Gaussian mixture model. In Section \ref{sec:filteringsmoothing}, we present the Quadrature Gaussian Sum filter and smoother for Wiener systems. The implementation aspects of the proposed approach are considered in Section \ref{sec:implementation}. Section \ref{sec:simulations} presents three numerical examples, and Section \ref{sec:conclusions} contains concluding remarks.

\textit{Notation}: All vectors and matrices are written in bold. The identity matrix of dimension $n\times n$ is denoted as $\mathbf{I}_n$. $\mathcal{N}(\mathbf{x};\bm{\mu},\bm{\Sigma})$ denotes a Gaussian PDF of the random variable $\mathbf{x}$ with mean $\bm{\mu}$ and covariance matrix $\bm{\Sigma}$, and $p(\mathbf{y}|\mathbf{x})$ refers to the conditional probability density or mass function of the random variable $\mathbf{y}$ given $\mathbf{x}$. The probability that the random variable $\mathbf{x}$ has a particular value $\mathbf{x}^*$ is denoted by $\mathbb{P}(\mathbf{x}=\mathbf{x}^*)$. If $g$ is a scalar-valued function, the preimage of the set $A\subseteq \mathbb{R}$ under $g$ is denoted by $g^{-1}(A)$.

\section{System Setup and Problem Formulation}\label{sec:setup}
Consider the multi-input single-output, time-invariant, discrete-time system in state-space form given by (Fig. \ref{fig:ssdiagram}):
\begin{figure}	
	\centering
	\includegraphics[width=\linewidth]{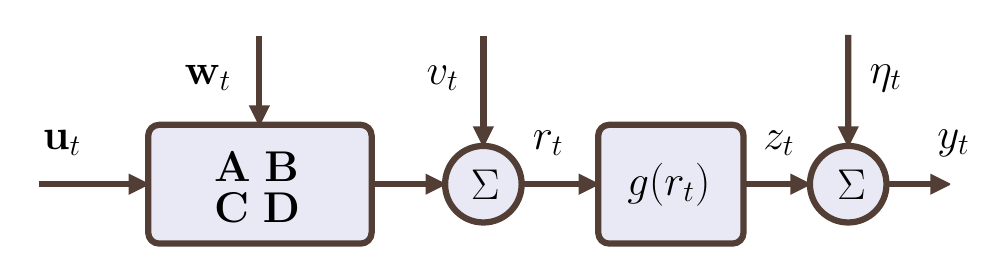}
 	\vspace{-6mm}
	\caption{Block diagram of a Wiener system in state-space form.}
	\label{fig:ssdiagram}
    	\vspace{-3mm}
\end{figure}
\begin{align}
    \mathbf{x}_{t+1}&=\mathbf{Ax}_{t}+\mathbf{B}\mathbf{u}_{t}+\mathbf{w}_{t},\label{eqn:general_system}\\
    r_{t}&=\mathbf{Cx}_{t}+\mathbf{D}\mathbf{u}_{t}+v_{t},\label{eqn:general_system2}\\
    z_{t}&=g(r_{t}),\label{eqn:general_system3}\\
    y_{t}&= z_{t}+\eta_t,\label{eqn:general_system4} 
\end{align}
where $\mathbf{u}_{t} \in \mathbb{R}^{m}$, $\mathbf{x}_{t} \in \mathbb{R}^{n}$, $r_{t} \in \mathbb{R}$, $z_{t} \in \mathbb{R}$, and $y_{t} \in \mathbb{R}$, are the system input, the state vector, the linear output (not physically measurable), the nonlinear output without measurement noise (not physically measurable), the nonlinear output with measurement noise (physically measurable), and the system input, respectively. The system matrices have dimensions $\mathbf{A} \in \mathbb{R}^{n\times n}$, $\mathbf{B} \in \mathbb{R}^{n\times m}$, $\mathbf{C} \in \mathbb{R}^{1\times n}$ and $\mathbf{D} \in \mathbb{R}^{1\times m}$. The process noise $\mathbf{w}_{t} \in \mathbb{R}^{n}$, the linear output noise $v_{t} \in \mathbb{R}$, and the nonlinear output noise $\eta_{t} \in \mathbb{R}$ are jointly independent Gaussian-distributed stochastic processes, with zero mean and covariance matrices $\mathbf{Q}$, $R$, and $P$ respectively. The static function $g(\cdot)\colon \mathbb{R}\rightarrow \mathbb{R}$ is a known mapping that is piecewise continuous and differentiable, with a countable number of discontinuities. An example of such function is presented in Fig. \ref{fig:function_g}. Note that we admit intervals where $g$ is constant; the sets in the domain of $g$ where this occurs are called quantization sets in this paper.
\begin{figure}
	\centering
	\includegraphics[width=1\linewidth]{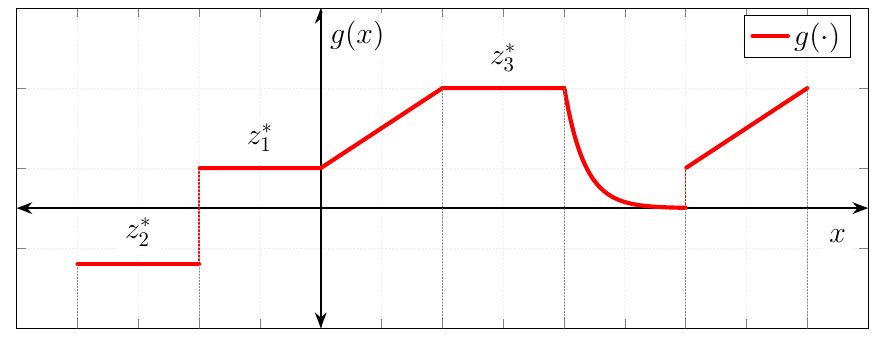}
 	\vspace{-6mm}
	\caption{An arbitrary piecewise nonlinear function $g$ admitted by our approach.}
	\label{fig:function_g}
    	\vspace{-3mm}
\end{figure}

The system in \eqref{eqn:general_system}-\eqref{eqn:general_system4} can be equivalently described by the following state transition probability density
\begin{equation}\label{eqn:prob_model_pxtm1_xt}
    p(\mathbf{x}_{t+1}|\mathbf{x}_t)=\mathcal{N}(\mathbf{x}_{t+1};\mathbf{Ax}_t+\mathbf{B}\mathbf{u}_t,\mathbf{Q}),
\end{equation}
with initial condition $\mathbf{x}_1$ drawn by $p(\mathbf{x}_1)=\mathcal{N}(\mathbf{x}_1;\bm{\mu}_1,\mathbf{P}_1)$, and by the linear and nonlinear output conditional PDFs
\begin{align}
    p(r_t|\mathbf{x}_t)&=\mathcal{N}(r_{t};\mathbf{Cx}_t+\mathbf{D}\mathbf{u}_t,R), \\
    p(y_t|r_t)&=\mathcal{N}(y_{t};g(r_t),P).
\end{align}
The problem addressed in this paper is how to obtain the PDFs of the state for filtering, $p(\mathbf{x}_t|y_{1:t})$, and smoothing, $p(\mathbf{x}_t|y_{1:N})$ and $p(\mathbf{x}_{t+1},\mathbf{x}_t|y_{1:N})$, from the available input and output data. For the filtering problem, the available data consists of $\mathbf{u}_{1:t}=\{\mathbf{u}_{1},\mathbf{u}_{2},\dots, \mathbf{u}_{t}\}$ and $y_{1:t}=\{y_{1},y_{2},\dots, y_{t}\}$. For the smoothing problem, the complete data set, $\mathbf{u}_{1:N}$ and $y_{1:N}$, is used, where $N$ is the total number of data points.

\section{Bayesian Filtering and Smoothing Recursions and a Quadrature-based Approach} \label{sec:gmm}
The goal of this section is to introduce the Bayesian filtering and smoothing recursions and to derive a Gaussian mixture model for the output PDF of a Wiener system.

\subsection{Bayesian Filtering and Smoothing}

Bayesian filtering involves a 2-step recursive procedure \cite{sarkka2013bayesian}, given by
\begin{align}
    p(\mathbf{x}_t|y_{1:t})&=\dfrac{p(y_t|\mathbf{x}_t)p(\mathbf{x}_t|y_{1:t-1})}{p(y_t|y_{1:t-1})},\label{eqn:bayesian_filtering_meas}\\			p(\mathbf{x}_{t+1}|y_{1:t})&=\int_{\mathbb{R}^n}p(\mathbf{x}_{t+1}|\mathbf{x}_t)p(\mathbf{x}_t|y_{1:t})\textnormal{d}\mathbf{x}_t, \label{eqn:bayesian_filtering_time}
\end{align}
where $p(\mathbf{x}_t|y_{1:t})$ and $p(\mathbf{x}_{t+1}|y_{1:t})$ correspond to the measurement and time updates, respectively, $p(y_t|\mathbf{x}_t)$ is the PDF of the nonlinear output given the state vector, $p(\mathbf{x}_{t+1}|\mathbf{x}_t)$ is the state transition PDF, and $p(y_t|y_{1:t-1})$ is a normalizing constant. On the other hand, in the Bayesian framework, the equation for nonlinear smoothing $p(\mathbf{x}_{t}|y_{1:N})$ is given by
\begin{equation}\label{eqn:bad_smoothing}
    p(\mathbf{x}_{t}|y_{1:N}) = p(\mathbf{x}_{t}|y_{1:t}) \int_{\mathbb{R}^n} \frac{p(\mathbf{x}_{t+1}|y_{1:N})p(\mathbf{x}_{t+1}|\mathbf{x}_{t})}{p(\mathbf{x}_{t+1}|y_{1:t})}\textnormal{d}\mathbf{x}_{t+1}.
\end{equation}
In the general case, when these PDFs are non-Gaussian, computing the integral in \eqref{eqn:bad_smoothing} becomes a difficult or intractable task. An alternative way to solve the smoothing problem is provided by the two-filter formula in \cite{kitagawa1994two}. In this approach, the smoothing equation is
\begin{equation} \label{eqn:bayesian_smoothing}
    p(\mathbf{x}_t|y_{1:N}) = \frac{p(\mathbf{x}_t|y_{1:t-1})p(y_{t:N}|\mathbf{x}_t)}{p(y_{t:N}|y_{1:t-1})},
\end{equation}
where $p(\mathbf{x}_t|y_{1:t-1})$ is the prediction PDF from the filtering stage, $p(y_{t:N}|y_{1:t-1})$ is a normalization constant, and $p(y_{t:N}|\mathbf{x}_t)$ is obtained from the following reverse recursion:
\begin{align}
    p(y_{t+1:N}|\mathbf{x}_t) \hspace{-0.04cm}&= \hspace{-0.08cm} \int_{\mathbb{R}^n}\hspace{-0.05cm} p(y_{t+1:N}|\mathbf{x}_{t+1})p(\mathbf{x}_{t+1}|\mathbf{x}_t)\textnormal{d}\mathbf{x}_{t+1}, \label{eqn:bayesian_backward_prediction}\\
    p(y_{t:N}|\mathbf{x}_t) \hspace{-0.04cm} &= \hspace{-0.04cm} p(y_t|\mathbf{x}_t)p(y_{t+1:N}|\mathbf{x}_t). \label{eqn:bayesian_backward_measurement}
\end{align}
Here, $p(y_{t+1:N}|\mathbf{x}_t)$ and $p(y_{t:N}|\mathbf{x}_t)$ are the reverse prediction and update equations, respectively. 
Additionally, the PDF of the state in two consecutive time instants is given by:
\begin{equation}\label{eqn:bayes_joint_smoothing}
    p(\mathbf{x}_{t+1},\mathbf{x}_t|y_{1:N}) \propto p(y_{t+1:N}|\mathbf{x}_{t+1})p(\mathbf{x}_{t+1}|\mathbf{x}_t)p(\mathbf{x}_t|y_{1:t}).
\end{equation}
Notice that the smoothing equation in \eqref{eqn:bayesian_smoothing} requires the prediction PDFs $p(\mathbf{x}_t|y_{1:t-1})$, which are obtained from \eqref{eqn:bayesian_filtering_meas} and~\eqref{eqn:bayesian_filtering_time}.

As observed in Eqs. \eqref{eqn:bayesian_filtering_meas} and \eqref{eqn:bayesian_filtering_time}, corresponding to the forward filtering stage, and in Eqs. \eqref{eqn:bayesian_backward_prediction} and \eqref{eqn:bayesian_backward_measurement} in the backward filtering stage, it is necessary to know the probability functions $p(\mathbf{x}_{t+1}|\mathbf{x}_t)$ and $p(y_{t}|\mathbf{x}_t)$. The equation for $p(\mathbf{x}_{t+1}|\mathbf{x}_t)$ can be obtained directly from (2). However, since the output $y_t$ is obtained by a signal affected by the possibly non-monotonic nonlinearity $g(\cdot)$, obtaining $p(y_{t}|\mathbf{x}_t)$ is not straightforward. Therefore, this section provides some developments that will allow us to obtain $p(y_{t}|\mathbf{x}_t)$ in such a way that we can solve the filtering and smoothing equations in closed-form. Note that because the nonlinear function $g(\cdot)$ can contain both strictly monotonic and constant sections, $p(z_t|\mathbf{x}_t)$ can be a PDF if all the pieces are strictly monotonic functions, a probability mass function (PMF) if all pieces are constants, or a generalized PDF (GPDF) if the pieces are a combination of strictly monotonic functions and constants. 

\subsection{Gauss-Legendre Quadrature rule for computing $p(y_t|\mathbf{x}_t)$} 
A critical step towards computing $p(y_t|\mathbf{x}_t)$ for a general Wiener state-space model involves transforming the random variable $r_t$ through the function $g(\cdot)$. The following lemma provides the general expression for such a type of GPDF.

\begin{lemma}
\label{lem:pz}
    Assume that the domain of $g(\cdot)$ can be partitioned into $M_1+M_2$ sets, of which $M_1$ correspond to intervals in which it is a differentiable and strictly monotonic function, and $M_2$ correspond to quantization sets of the strictly positive Lebesgue measure in which $g(\cdot)$ is constant and distinct between sets. Let $z_1^*, \dots, z_{M_2}^*\in \mathbb{R}$ be the image of each quantization set. Define $z = g(r)$, where $r\sim\mathcal{N}(r;\mu,R)$, and denote $\gamma_i(z)$ as the $i$th root ($i=1,\dots,K_{z}$) of the equation $z=g(r)$, i.e., $g(\gamma_i(r))=r$. Then 
    \begin{align}
        p(z) &= \sum_{i=1}^{M_1}  \phi_i(z)\mathcal{N}(\tilde{\gamma}_i(z);\mu,R) \nonumber \\
        &+ \sum_{j=1}^{M_2} \delta(z-z_j^*) \int_{r\in g^{-1}(z_j^*)}\mathcal{N}(r;\mu,R)\textnormal{d}r, \label{equationlemmaZ}
    \end{align}
    where $\delta(\cdot)$ is the Dirac delta function, and 
     \begin{equation}
        \label{gammatildelemma}
        \tilde{\gamma}_i(z)= \begin{cases}
        \gamma_i(z) & \textnormal{if } i\leq K_z, \\
        0 & \textnormal{otherwise},
    \end{cases}, \quad
    \phi_i(z)=\bigg|\frac{\textnormal{d}\tilde{\gamma}_i(z)}{\textnormal{d}z}\bigg|. \end{equation}
\end{lemma}
\begin{proof}
    See Appendix \ref{appendix:lem:pz}.
\end{proof}
Lemma \ref{lem:pz} enables the derivation of the exact expression for the nonlinear output conditional PDF when the signal is contaminated by white Gaussian noise after the nonlinearity block, as seen in Theorem \ref{thm:pytxt}. Note that there exists a class of systems in which the output measurements are only perturbed by noise prior to the nonlinear static block. Specifically, in systems with quantized output measurements, the output is typically represented as $y_t=z_t=g(r_t)$, where $g$ denotes the quantizer. The lemma \ref{lem:pz} allows the direct computation of the probability function $p(z_t|\mathbf{x}_t)$.

\begin{theorem}\label{thm:pytxt}
	Consider the system in \eqref{eqn:general_system}-\eqref{eqn:general_system4}, and assume that $g$ satisfies the conditions in Lemma \ref{lem:pz}. Then,
    \begin{align}
        \hspace{-0.2cm}p(y_t|\mathbf{x}_t)&= \notag \\
        &\hspace{-1.3cm}\sum_{i=1}^{M_1}\hspace{-0.06cm}\int_{\mathbb{R}}\hspace{-0.07cm} \phi_{i}(y_t\hspace{-0.07cm}-\hspace{-0.05cm}\eta_t)\mathcal{N}\hspace{-0.04cm}\left(\tilde{\gamma}_{i}(y_t\hspace{-0.07cm}-\hspace{-0.05cm}\eta_t);\hspace{-0.02cm}\mathbf{Cx}_t\hspace{-0.07cm}+\hspace{-0.06cm}\mathbf{D}\mathbf{u}_t,\hspace{-0.03cm}R\right)\mathcal{N}\hspace{-0.04cm}\left(\eta_t;\hspace{-0.02cm}0,\hspace{-0.03cm}P\right)\hspace{-0.03cm} \textnormal{d}\eta_t \nonumber\\
        &\hspace{-1.2cm}+ \hspace{-0.05cm}\sum_{j=1}^{M_2} \hspace{-0.03cm}\mathcal{N}(y_t\hspace{-0.05cm}-\hspace{-0.05cm}z_j^*;0,\hspace{-0.03cm}P)\hspace{-0.09cm}\int_{r_t\in g^{-1}(z_j^*)}\hspace{-0.7cm}\mathcal{N}(r_t;\mathbf{Cx}_t\hspace{-0.05cm}+\hspace{-0.05cm}\mathbf{D}\mathbf{u}_t,R)\textnormal{d}r_t, \label{eqn:integral_lemma_pytxt}
    \end{align}
	where $\tilde{\gamma}_i(z_t), \phi_i(z_t)$ are defined in \eqref{gammatildelemma}.
\end{theorem}
\begin{proof}
    See Appendix \ref{appendix:thm:pytxt}.
\end{proof}

Notice that the integrals in \eqref{eqn:integral_lemma_pytxt} cannot be computed in closed form, in general. Hence, in this work we propose to approximate these integrals in a convenient manner, leading to explicit filtering and smoothing formulas. For the sake of brevity, we assume that the preimages of $g(z_j^*)$ are bounded intervals, that is, $g^{-1}(z_j^*) = [\underline{q}_j, \bar{q}_j]$, where $\underline{q}_j$ and $\bar{q}_j$ are the lower and upper limits of the interval. General expressions can be derived at the expense of a more involved notation.

Consider the following Gauss-Legendre quadrature rule for approximating an integral in the finite interval $[-1,1]$:
\begin{equation}
    \int_{-1}^1 h(x)\textnormal{d}x = \sum_{\tau=1}^L \omega_\tau h(\psi_\tau)+R_L,
\end{equation}
where $\omega_\tau$ and $\psi_\tau$ are the weights of the quadrature and the roots of the Legendre polynomial of order $L$, respectively \cite{cohen2011numerical}, and $R_L$ is a residual term \cite{kahaner1989numerical}. This quadrature rule can be adapted to the intervals $[\underline{q}_{j},\bar{q}_j]$ and $(-\infty,\infty)$ by considering the changes of variables $x=\varepsilon (\bar{q}_j-\underline{q}_{j})/2 + (\bar{q}_j+\underline{q}_{j})/2$ and $x=\varepsilon/(1-\varepsilon^2)$ respectively, which give 
\begin{align}\label{eqn:approx_finite}
    \int_{\underline{q}_{j}}^{\bar{q}_j}\hspace{-0.07cm}h(x)\textnormal{d}x &\approx \hspace{-0.05cm}\sum_{\tau=1}^{L}\hspace{-0.03cm}\omega_{\tau}h\hspace{-0.05cm}\left(\frac{\bar{q}_j\hspace{-0.05cm}-\hspace{-0.05cm}\underline{q}_{j}}{2} \psi_{\tau}\hspace{-0.05cm}+\hspace{-0.05cm} \frac{\bar{q}_j\hspace{-0.05cm}+\hspace{-0.05cm}\underline{q}_{j}}{2}\right)\frac{\bar{q}_j\hspace{-0.05cm}-\hspace{-0.05cm}\underline{q}_{j}}{2}, \\
\label{eqn:approx_at}
    \int_{-\infty}^{\infty}h(x)\textnormal{d}x &\approx \sum_{\tau=1}^{L}\omega_{\tau}h\left(\frac{\psi_{\tau}}{1-\psi_{\tau}^2}\right)\frac{1+\psi_{\tau}^2}{(1-\psi_{\tau}^2)^2}.
\end{align}
Then, applying these approximations to $p(y_t|\mathbf{x}_t)$ yields
\begin{align}
    p(y_t|\mathbf{x}_t) &\approx \sum_{\ell_1=1}^{K_1} \beta_{\ell_1}\mathcal{N}(s_{\ell_1},\mathbf{Cx}_t+\mathbf{D}\mathbf{u}_t,R), \nonumber\\
    &+ \sum_{\ell_2=1}^{K_2}\alpha_{\ell_2}\mathcal{N}(m_{\ell_2},\mathbf{Cx}_t+\mathbf{D}\mathbf{u}_t,R), \label{eqn:y_approximation}
\end{align}
where $K_1=M_1L_1$, $K_2=M_2L_2$, $\ell_1=(i-1)L_1+\tau_1$, $\ell_2=(j-1)L_2+\tau_2$, and
\begin{align}
         s_{\ell_1}&=\tilde{\gamma}_i(y_t-\lambda_{\tau_1}),\\
         m_{\ell_2}\hspace{-0.06cm}&=\psi_{\tau_2}(\bar{q}_j-\underline{q}_j)/2+(\bar{q}_j+\underline{q}_j)/2,\\
        \beta_{\ell_1}&=\hspace{-0.04cm}\omega_{\tau_1}\phi_i(y_t\hspace{-0.07cm}-\hspace{-0.07cm}\lambda_{\tau_1}\hspace{-0.01cm}) \mathcal{N}\hspace{-0.04cm}\left(\lambda_{\tau_1}\hspace{-0.02cm};0,\hspace{-0.03cm}P\right)\hspace{-0.04cm}(1\hspace{-0.06cm}+\hspace{-0.04cm}\psi_{\tau_1}^2)/(1\hspace{-0.07cm}-\hspace{-0.04cm}\psi_{\hspace{-0.01cm}\tau_1}^2)^2\hspace{-0.02cm},\hspace{-0.04cm}\\
        \alpha_{\ell_2} &= \omega_{\tau_2}\mathcal{N}(y_t-z_j^*;0,P)(\bar{q}_j-\underline{q}_j)/2, \label{eqn:alpha_elle2}
\end{align}
where $\tau_1$ and $\tau_2$ are integers that range from $1$ to $L_1$ and $1$ to $L_2$, respectively, and where $\lambda_{\tau_1}=\psi_{\tau_1}/(1-\psi_{\tau_1}^2)$. Notice that $p(y_t|\mathbf{x}_t)$ can be rewritten as a single summation defining a new index $\kappa=1,\dots,K_1,K_1+1,\dots,K_2$, i.e., 
\begin{align}\label{eqn:y_approximation_rew}
    p(y_t|\mathbf{x}_t) &\approx \sum_{\kappa=1}^{K} \varphi_t^{\kappa}\mathcal{N}(\zeta_t^{\kappa}(y_t);\mathbf{Cx}_t+\mathbf{D}\mathbf{u}_t,R),
\end{align}
where the total number of Gaussian components is $K=K_1+K_2$, and the pair $\varphi^{\kappa},\zeta^{\kappa}(y_t)$ takes values following
\begin{align}
    \varphi_t^{\kappa},\zeta_t^{\kappa}(y_t) &= \left\lbrace 
    \begin{matrix}
    	\beta_{\kappa},s_{\kappa} & \textnormal{\hspace{-0.07cm}if}  & \hspace{-0.04cm}1 \leq \kappa \leq K_1, \\
    	\alpha_{\kappa-K_1},m_{\kappa-K_1} & \textnormal{\hspace{-0.07cm}if}  & \hspace{-0.04cm}K_1\hspace{-0.04cm} + \hspace{-0.04cm}1  \leq \kappa \leq K.
    \end{matrix}
    \right. \hspace{-0.05cm}\label{eqn:elements_app_py}
\end{align}

\section{The Quadrature Gaussian Sum Filter and Smoother for Wiener Systems}\label{sec:filteringsmoothing}
Theorem \ref{thm:pytxt} defines an explicit Gaussian mixture model for $p(y_t|\mathbf{x}_t)$, which is now exploited to design the proposed Quadrature Gaussian sum filtering (QGSF) and smoothing (QGSS) algorithms for Wiener state-space systems. 

\subsection{Quadrature Gaussian Sum Filtering for Wiener systems}
The QGSF algorithm for Wiener state-space systems is summarized in the following theorems:
  
\begin{theorem}\label{thm:filtering}
    Consider the system in \eqref{eqn:general_system}-\eqref{eqn:general_system4}, the PDF $p(y_t|\mathbf{x}_t)$ in Theorem \ref{thm:pytxt}, and its approximation in \eqref{eqn:y_approximation_rew}. Then, the Quadrature Gaussian Sum Filter for Wiener state-space systems is given as follows. For time $t=1$, the PDF of the time update step is given by $p(\mathbf{x}_1)=\mathcal{N}(\mathbf{x}_1;\bm{\mu}_1,\mathbf{P}_1)$, and for $t=1,\dots,N$, the following steps are defined:\\
    
    \noindent \textbf{Measurement Update:} The filtered PDF of the current state $\mathbf{x}_t$ given measurements of the nonlinear output $y_1,\dots,y_t$ is the Gaussian mixture model
    \begin{equation}\label{eqn:correction_filtering} 
        p(\mathbf{x}_t|y_{1:t})=\sum_{k=1}^{M_{t|t}} \delta_{t|t}^{k}\mathcal{N}(\mathbf{x}_{t};\hat{\mathbf{x}}_{t|t}^{k},\bm{\Sigma}_{t|t}^{k}),
    \end{equation}
    where for each pair $(\kappa,\ell)$, with $\kappa=1,\dots,K$ and $\ell=1,\dots, M_{t|t-1}$, a new index $k=(\ell-1)K+\kappa$ is obtained, such that the weights, means, and covariances are given by
    \begin{align}
            M_{t|t}&=KM_{t|t-1}, \label{eqn:lemma_filtering_0}\\
            \delta_{t|t}^{k} & \!=\! \bar{\delta}_{t|t}^{k}/\textstyle\sum_{r=1}^{M_{t|t}}\bar{\delta}_{t|t}^{r}, \label{eqn:lemma_filtering_1}\\	
            \bar{\delta}_{t|t}^{k}&\!=\!\varphi_t^{\kappa} \delta_{t|t-1}^{\ell}\mathcal{N}\hspace{-0.05cm}\prt{\hspace{-0.02cm}\zeta_t^{\kappa}(y_t);\mathbf{C}\hat{\mathbf{x}}_{t|t\hspace{-0.01cm}-\hspace{-0.01cm}1}^{\ell}\hspace{-0.07cm}+\hspace{-0.05cm}\mathbf{D}\mathbf{u}_t,R\hspace{-0.05cm}+\hspace{-0.05cm}\mathbf{C}\bm{\Sigma}_{t|t\hspace{-0.01cm}-\hspace{-0.01cm}1}^{\ell}\hspace{-0.02cm}\mathbf{C}^{\hspace{-0.02cm}\top}\hspace{-0.02cm}}\hspace{-0.05cm},\label{eqn:lemma_filtering_2}\\
            \mathbf{K}_t^{\ell}&\!=\!\bm{\Sigma}_{t|t-1}^{\ell}\mathbf{C}^{\top}(R\hspace{-0.05cm}+\hspace{-0.05cm}\mathbf{C}\bm{\Sigma}_{t|t-1}^{\ell}\mathbf{C}^{\top})^{-1}, \label{eqn:lemma_filtering_3}\\ 
            \hat{\mathbf{x}}_{t|t}^{k}&\!=\!\hat{\mathbf{x}}_{t|t-1}^{\ell}+\mathbf{K}_t^{\ell}(\zeta_t^{\kappa}(y_t)-\mathbf{C}\hat{\mathbf{x}}_{t|t-1}^{\ell}-\mathbf{D}\mathbf{u}_t),\label{eqn:lemma_filtering_4}\\
            \bm{\Sigma}_{t|t}^{k}&\!=\!(\mathbf{I}_n-\mathbf{K}_t^{\ell}\mathbf{C})\bm{\Sigma}_{t|t-1}^{\ell},\label{eqn:lemma_filtering_5}
    \end{align}
    where $K$, $\varphi_t^{\kappa}$, and $\zeta_t^{\kappa}(y_t)$ describe the Gaussian mixture model in \eqref{eqn:y_approximation_rew}. The initial values of the recursion are $M_{1|0}=1$, $\delta_{1|0}=1$, $\hat{\mathbf{x}}_{1|0}=\bm{\mu}_1$, and $\bm{\Sigma}_{1|0}=\mathbf{P}_1$.\\
    
    \noindent \textbf{Time Update:} The predicted PDF of the state $\mathbf{x}_{t+1}$ given measurements of the nonlinear output $y_1,\dots,y_t$ is given by
    \begin{equation}\label{eqn:prediccion_filtering}
        p(\mathbf{x}_{t+1}|y_{1:t})=\sum_{k=1}^{M_{t+1|t}} \delta_{t+1|t}^{k}\mathcal{N}(\mathbf{x}_{t+1};\hat{\mathbf{x}}_{t+1|t}^{k},\bm{\Sigma}_{t+1|t}^{k}),
    \end{equation}
    where $M_{t+1|t}=M_{t|t}, \delta_{t+1|t}^{k}=\delta_{t|t}^{k}$, and
    \begin{align}
\hat{\mathbf{x}}_{t+1|t}^{k}&=\mathbf{A}\hat{\mathbf{x}}_{t|t}^{k}+\mathbf{B}\mathbf{u}_t, \label{eqn:lemma_filtering_7}\\	
        \bm{\Sigma}_{t+1|t}^{k}&=\mathbf{Q}+\mathbf{A}\bm{\Sigma}_{t|t}^{k}\mathbf{A}^\top. \label{eqn:lemma_filtering_8}
    \end{align}
\end{theorem}
\begin{proof}
	See Appendix \ref{proof:filtering}.
\end{proof}
By exploiting the approximation of $p(y_t|\mathbf{x}_t)$ defined in \eqref{eqn:y_approximation_rew}, the backward filtering algorithm is stated next.
\begin{theorem}\label{thm:backward}
	Consider the system in \eqref{eqn:general_system}-\eqref{eqn:general_system4}, the PDF $p(y_t|\mathbf{x}_t)$ given in Theorem \ref{thm:pytxt}, and its approximation in \eqref{eqn:y_approximation_rew}. Then, the backward filtering algorithm for Wiener systems is given as follows: For $t=N$, the equation for the correction stage is
	\begin{align}
	    p(y_{N}|\mathbf{x}_N) \!&=\! \sum_{k=1}^{S_{N|N}} \!\! \epsilon_{N|N}^{k}  \mathcal{N}\big( \zeta_{N:N}^{k}(y_{N:N});\nonumber\\
        &\left. \mathcal{O}_{N|N}\mathbf{x}_N+\mathcal{H}_{N|N}\mathbf{u}_{N:N},\mathcal{P}_{N|N}\right),\label{eqn:backward_tN}
	\end{align}
	where $S_{N|N}=K$, $\epsilon_{N|N}^{k}=\varphi_{N}^{k}$, $\mathcal{O}_{N|N}=\mathbf{C}$, $\mathcal{H}_{N|N}=\mathbf{D}$ and $\mathcal{P}_{N|N}=R$, where $\varphi_N^{k}$ and $\zeta_{N:N}^{k}(y_{N:N})=\zeta_{N}^{k}(y_{N})$ are defined in \eqref{eqn:elements_app_py}, and $\mathbf{u}_{N:N}=\mathbf{u}_N$. \\
	
	\noindent\textbf{Backward prediction:} For $t=N-1,\dots,1$ the backward prediction equation is defined as
	\begin{align}\label{eqn:backward_prediction}
		p(y_{t+1:N}|\mathbf{x}_t)&\!=\!\!\!\sum_{k=1}^{S_{t|t+1}}\!\!\epsilon_{t|t+1}^{k} \mathcal{N}\left(\zeta_{t+1:N}^{k}(y_{t+1:N});\right.\nonumber\\
        & \left.\mathcal{O}_{t|t+1}\mathbf{x}_t+\mathcal{H}_{t|t+1}\mathbf{u}_{t:N}, \mathcal{P}_{t|t+1}\right),
	\end{align}
	where $S_{t|t+1}=S_{t+1|t+1}$, $\epsilon_{t|t+1}^{k}=\epsilon_{t+1|t+1}^{k}$, and 
		\begin{align}
            \label{eqn:backward_prediction_S}
            \hspace{-0.5cm}\mathcal{O}_{t|t+1}\hspace{-0.08cm}&=\hspace{-0.05cm}\mathcal{O}_{t+1|t+1}\mathbf{A}, \hspace{-0.05cm}\quad \hspace{-0.12cm}\mathcal{H}_{t|t\hspace{-0.02cm}+\hspace{-0.02cm}1}\hspace{-0.07cm}=\hspace{-0.09cm}\begin{bmatrix}\mathcal{O}_{t\hspace{-0.02cm}+\hspace{-0.02cm}1|t\hspace{-0.02cm}+\hspace{-0.02cm}1}\mathbf{B},& \hspace{-0.22cm}\mathcal{H}_{t\hspace{-0.02cm}+\hspace{-0.02cm}1|t\hspace{-0.02cm}+\hspace{-0.02cm}1}\end{bmatrix}\hspace{-0.03cm}, \\
            \hspace{-0.5cm}\mathcal{P}_{t|t+1}&=\mathcal{P}_{t+1|t+1} + \mathcal{O}_{t+1|t+1} \mathbf{Q} \mathcal{O}_{t+1|t+1}^\top. \label{eqn:backward_prediction_P}	
	\end{align}
    
 \noindent\textbf{Backward Measurement Update:} For $t=N-1,\dots,1$, and for each pair $(\tau,k)$, where $\tau=1,\dots,K$ and $k=1,\dots,S_{t|t+1}$, let $\ell$ be a new index $\ell=(k-1)K+\tau$. Then, 	\begin{align}\label{eqn:backward_measureement}
		p(y_{t:N}|\mathbf{x}_t)&=\sum_{\ell=1}^{S_{t|t}}\hspace{-0.03cm}\epsilon_{t|t}^{\ell}\mathcal{N}\prt{\zeta_{t:N}^{\ell}(y_{t:N}); \mathcal{O}_{t|t}\mathbf{x}_t+\mathcal{H}_{t|t}\mathbf{u}_{t:N}, \mathcal{P}_{t|t}}\hspace{-0.03cm},
	\end{align}
	where $S_{t|t}=KS_{t|t+1}$, $\epsilon_{t|t}^{\ell}=\varphi_{t}^{\tau}\epsilon_{t|t+1}^{k}$, and
	\begin{align}
            \label{eqn:backward_measurement_S}
            \zeta_{t:N}^{\ell}(y_{t:N})&=\begin{bmatrix}\zeta_{t}^{\tau}(y_{t}) \\ \zeta_{t+1:N}^{k}(y_{t+1:N})\end{bmatrix}, \quad \hspace{-0.08cm} \mathcal{O}_{t|t}=\begin{bmatrix} \mathbf{C} \\ \mathcal{O}_{t|t+1} \end{bmatrix}\hspace{-0.05cm}, \\
            \mathcal{H}_{t|t}&=\begin{bmatrix} \hspace{0.03cm}\mathbf{D} \hspace{0.48cm} \mathbf{0}^\top \\ \mathcal{H}_{t|t+1} \end{bmatrix}, \quad \mathcal{P}_{t|t}=\begin{bmatrix} R & \mathbf{0}^\top \\ \mathbf{0} & \mathcal{P}_{t|t+1} \end{bmatrix}.
            \label{eqn:backward_measurement_H}
	\end{align}
\end{theorem}
\begin{proof}
	See Appendix \ref{proof:backward}. 
\end{proof}

\subsection{Quadrature Gaussian Sum Smoothing for Wiener systems}

\begin{theorem} \label{thm:smoothing}
Consider the system in \eqref{eqn:general_system}-\eqref{eqn:general_system4}, the PDF $p(y_t|\mathbf{x}_t)$ given in Theorem \ref{thm:pytxt}, and its approximation in \eqref{eqn:y_approximation_rew}. Given $p(\mathbf{x}_t|y_{1:t-1})$, $p(\mathbf{x}_N|y_{1:N})$ and $p(y_{t:N}|\mathbf {x}_t)$, the smoothing algorithm for the nonlinear data set $y_{1:N}$ of a Wiener model in state-space form is as follows: \\
The smoothing PDF at time $t=N$ is $p(\mathbf{x}_N|y_{1:N})$, which corresponds to the PDF of the last iteration of the filtering algorithm in the measurement stage, and for $t=N-1,\dots,1$, the smoothing PDF $p(\mathbf{x}_{t}|y_{1:N})$ is given by:
\begin{equation}\label{eqn:smoothing_1_wiener}
    p(\mathbf{x}_t|y_{1:N})=\sum_{k=1}^{S_{t|N}}\delta_{t|N}^{k}\mathcal{N}(\mathbf{x}_t;\hat{\mathbf{x}}_{t|N}^{k},\bm{\Sigma}_{t|N}^{k}),
\end{equation}
where $k$ is an index transforming $(\tau,\ell)$ by $k \hspace{-0.05cm}=\hspace{-0.05cm}(\ell\hspace{-0.05cm}-\hspace{-0.05cm}1)M_{t|t\hspace{-0.02cm}-\hspace{-0.02cm}1}\hspace{-0.02cm}+\hspace{-0.02cm}\tau$ with $\tau=1,\dots,M_{t|t-1}$ and $\ell=1,\dots,S_{t|t}$ and where:\\
\begin{align}
    \label{eqsmoothingS}
    S_{t|N}\hspace{-0.07cm}&=M_{t|t-1}S_{t|t},\\
    \label{eqsmoothingdeltabar}
    \delta_{t|N}^{k}\hspace{-0.07cm}&=\bar{\delta}_{t|N}^{k}/\textstyle\sum_{s=1}^{S_{t|N}}\bar{\delta}_{t|N}^{s},\\
    \bar{\delta}_{t|N}^{k}\hspace{-0.07cm}&= \delta_{t|t-1}^{\tau}\epsilon_{t|t}^{\ell} \notag \\
    \label{eqsmoothingdelta}
    & \hspace{-0.8cm} \times\hspace{-0.05cm} \mathcal{N}\hspace{-0.03cm}\!\prt{\hspace{-0.04cm}\zeta_{t:N}^{\ell}\hspace{-0.02cm}(y_{t:N});\hspace{-0.02cm}\mathcal{O}_{t|t}\hat{\mathbf{x}}_{t|t\hspace{-0.01cm}-\hspace{-0.01cm}1}^{\tau}\!\!+\!\mathcal{H}_{t|t}\mathbf{u}_{t:N},\hspace{-0.02cm}\mathcal{P}_{\hspace{-0.02cm}t|t}\!+\!\mathcal{O}_{t|t}\bm{\Sigma}_{\hspace{-0.03cm}t|t\hspace{-0.01cm}-\hspace{-0.01cm}1}^{\tau}\hspace{-0.02cm}\mathcal{O}_{t|t}^{\top}\hspace{-0.03cm}}\hspace{-0.08cm}, \\
    \label{eqsmoothingx}
    \hat{\mathbf{x}}_{t|N}^{k}\hspace{-0.08cm}&=\hspace{-0.05cm}\hat{\mathbf{x}}_{t|t\hspace{-0.02cm}-\hspace{-0.02cm}1}^{\tau}\!\hspace{-0.06cm}+\hspace{-0.01cm}\!\mathbf{K}_{t|N}^{\tau}\hspace{-0.02cm}(\hspace{-0.01cm}\zeta_{t:N}^{\ell}\hspace{-0.02cm}(y_{t:N}\hspace{-0.01cm})\hspace{-0.07cm}-\hspace{-0.06cm}\mathcal{O}_{t|t}\hat{\mathbf{x}}_{t|t\hspace{-0.01cm}-\hspace{-0.01cm}1}^{\tau}\!\!\hspace{-0.03cm}-\hspace{-0.02cm}\!\mathcal{H}_{t|t}\mathbf{u}_{t:N}\hspace{-0.03cm})\hspace{-0.02cm}, \hspace{-0.1cm}\\
    \label{eqsmoothingsigma}
    \bm{\Sigma}_{t|N}^{k}\hspace{-0.07cm}&=(\mathbf{I}_n-\mathbf{K}_{t|N}^{\tau}\mathcal{O}_{t|t})\bm{\Sigma}_{t|t-1}^{\tau}, \\
    \label{eqsmoothingk}
    \mathbf{K}_{t|N}^{\tau}\hspace{-0.07cm}&=\bm{\Sigma}_{t|t-1}^{\tau}\mathcal{O}_{t|t}^{\top}(\mathcal{P}_{t|t}+\mathcal{O}_{t|t}\bm{\Sigma}_{t|t-1}^{\tau}\mathcal{O}_{t|t}^{\top})^{-1}.
\end{align}
\end{theorem}
\begin{proof} 
	See Appendix \ref{proof:smoothing}.
\end{proof}

In addition to the smoothing equation $p(\mathbf{x}_t|y_{1:N})$ in \eqref{eqn:smoothing_1_wiener}, some system identification algorithms require the computation of the joint PDF of the state at two consecutive instants in time $p(\mathbf{x}_{t+1},\mathbf{x}_t|y_{1:N})$, such as in \cite{gibson2005robust,schon2011,Wills2013,gonzalez2023algorithm}. In the case of Wiener systems, the proposed algorithm enables the computation of this PDF from Theorem \ref{thm:joint}.

\begin{theorem}\label{thm:joint}
    Consider the system in \eqref{eqn:general_system}-\eqref{eqn:general_system4}, the PDF $p(y_t|\mathbf{x}_t)$ given in Theorem \ref{thm:pytxt}, and its approximation in \eqref{eqn:y_approximation_rew}. Furthermore, consider the PDF $p(\mathbf{x}_t|y_{1:t})$ and the reverse prediction stage function $p(y_{t+1:N}|\mathbf{x}_{ t+1})$ given in Theorems \ref{thm:filtering} and \ref{thm:backward}, respectively, and the PDF $p(\mathbf{x}_{t+1}|\mathbf{x }_t)$ given in \eqref{eqn:prob_model_pxtm1_xt}. Then, defining the extended vector $(\mathbf{x}_{t}^{\text{e}})^{\top}=[\begin{matrix} \mathbf{x}_{t+1}^ {\top},&\mathbf{x}_t^{\top} \end{matrix}]^{\top}$, we have for $t=N-1,\dots,1$:
    \begin{equation}\label{eqn:thm_joint_smoothig_wiener}
        p(\mathbf{x}_{t+1},\mathbf{x}_t|\mathbf{y}_{1:N})=\sum_{k=1}^{S_{t+1|N}}\delta_{t+1|N}^{k}\mathcal{N}(\mathbf{x}_{t}^{\text{e}};\hat{\mathbf{x}}_{t|N}^{\text{e}(k)},\bm{\Sigma}_{t|N}^{\text{e}(k)}),
    \end{equation}
    where $S_{t+1|N}$ and $\delta_{t+1|N}^{k}$ given by $\eqref{eqsmoothingS}$ and \eqref{eqsmoothingdeltabar} respectively when evaluated at $t=t+1$, and the index $k$ transforms the pair $(\tau,\ell)$ by $k=(\ell-1)M_{t|t}+\tau$, with $\tau=1,\dots,M_ {t|t}$ and $\ell=1,\dots,S_{t+1|t+1}$. The means and covariances are given by
    \begin{align}
    \label{jointsmoothingx}
        \hat{\mathbf{x}}_{t|N}^{\text{e}(k)}\hspace{-0.11cm}&=\hspace{-0.125cm}\begin{bmatrix}
            \hat{\mathbf{x}}_{t+1|N}^k \\ \hat{\mathbf{x}}_{t|t}^\tau + \mathbf{K}_{t|t+1}^\tau (\hat{\mathbf{x}}_{t+1|N}^k-\hat{\mathbf{x}}_{t+1|t}^\tau)
        \end{bmatrix},\\
        \label{jointsmoothingsigma} 
        \bm{\Sigma}_{\hspace{-0.02cm}t|N}^{\text{e}(k)}\hspace{-0.11cm}&=\hspace{-0.125cm}\begin{bmatrix}
            \bm{\Sigma}_{t+1|N}^k & \hspace{-0.1cm}\bm{\Sigma}_{t+1\hspace{-0.005cm}|\hspace{-0.005cm}N}^k\mathbf{K}_{t\hspace{-0.005cm}|\hspace{-0.005cm}t+1}^{\tau\top} \\ \mathbf{K}_{t\hspace{-0.005cm}|\hspace{-0.005cm}t\hspace{-0.02cm}+\hspace{-0.02cm}1}^\tau \hspace{-0.02cm}\bm{\Sigma}_{\hspace{-0.02cm}t\hspace{-0.02cm}+\hspace{-0.02cm}1\hspace{-0.005cm}|\hspace{-0.005cm}N}^k &\hspace{-0.13cm} \bm{\Sigma}_{t\hspace{-0.005cm}|\hspace{-0.005cm}t}^\tau\hspace{-0.13cm}+\hspace{-0.07cm}\mathbf{K}_{\hspace{-0.02cm}t|t\hspace{-0.02cm}+\hspace{-0.02cm}1}^\tau \hspace{-0.03cm} (\bm{\Sigma}_{t\hspace{-0.02cm}+\hspace{-0.02cm}1\hspace{-0.005cm}|\hspace{-0.005cm}N}^k \hspace{-0.11cm}-\hspace{-0.07cm}\bm{\Sigma}_{t\hspace{-0.02cm}+\hspace{-0.02cm}1|t}^\tau )\mathbf{K}_{\hspace{-0.02cm}t\hspace{-0.005cm}|\hspace{-0.005cm}t\hspace{-0.02cm}+\hspace{-0.02cm}1}^{\tau\top} \end{bmatrix}\hspace{-0.07cm},
    \end{align}
    where we have defined $\mathbf{K}_{t|t+1}^{\tau}=\bm{\Sigma}_{t|t}^\tau \mathbf{A}^\top (\mathbf{Q}+\mathbf{A}\bm{\Sigma}_{t|t}^\tau \mathbf{A}^\top)^{-1}$. At instant $t=N$, the joint PDF $p(\mathbf{x}_{N+1},\mathbf{x}_N|y_{1:N})$ is given by \eqref{eqn:thm_joint_smoothig_wiener} evaluating $\ell=1$, that is, $\tau=k$, where $S_{N+1|N}=M_{N|N}$, and $\delta_{N+1|N}^k = \delta_{N|N}^{k}$ is calculated from \eqref{eqn:lemma_filtering_1}.
    \end{theorem}
\begin{proof}
	See Appendix \ref{proof:joint}.
\end{proof}

\section{Implementation Aspects and Relations with Established Methods}
\label{sec:implementation}
\subsection{Gaussian sum reduction}
In the measurement update step of Theorem \ref{thm:filtering}, and in the backward measurement step of Theorem \ref{thm:backward}, the number of Gaussian components at each iteration follows the equations $M_{t|t}=KM_{t|t-1}$ and $S_{t|t}=KS_{t|t+1}$, respectively, with $K\geq 1$. For $K>1$, this exponential growth renders the computational cost of the algorithm practically unmanageable. Therefore, as in other Gaussian sum filters \cite{wills2023numerically}, it is imperative to constrain this growth in each iteration to maintain the representation of filtered PDFs with a reduced number of Gaussian components. In this article, we employ the joining technique for Gaussian components reduction \cite{kitagawa1994two}, which combines the two Gaussians that minimize the similarity measure given by the Kullback-Leibler divergence.

For the case of backward filtering, the reduction can be implemented in a Gaussian mixture model as in \cite{runnalls2007kullback}, for which it is necessary to transform the expression of $p(y_{t:N}|\mathbf{x }_t)$ to a Gaussian mixture model. The component reduction is applied to this new model, and finally returns to the backward filter form in \eqref{eqn:backward_measureement}. 
The lemma \ref{lem:lemma_gmm_bif} suggests the procedure to switch between the Gaussian sum and the backward filter representation. Once the component reduction has been carried out in \eqref{eqn:backward_measureement}, a reduced version for $p(y_{t:N}|\mathbf{x}_t)$ is obtained as
\begin{equation}\label{eqn:reduced_gmm}
    p(y_{t:N}|\mathbf{x}_t) = \sum _{\ell=1}^{S_{t|t}^{\textnormal{red}}}\gamma _{t|t}^{\ell} \mathcal{N }\left(\mathbf{x}_t;\mathbf{z}_{t|t}^{\ell},\mathbf{U}_{t|t}^{\ell}\right),
\end{equation}
where $S_{t|t}^{\textnormal{red}}$ is the number of components after reduction and $\gamma _{t|t}^{\ell}$, $\mathbf{z}_{t|t}^{\ell}$, and $\mathbf{U}_{t|t}^{\ell}$ are the weights, means, and covariance matrices, respectively, of each Gaussian component.

\vspace{-0.2cm}
\subsection{State estimator based on Gaussian mixture models}
Once the filtering and smoothing PDFs are calculated as indicated in Theorems \ref{thm:filtering}, \ref{thm:smoothing}, and \ref{thm:joint}, the estimator of the system states based on data $y_{1:t}$ and $y_{1:N}$, as well as their respective covariance matrix of the estimation error, are calculated as follows. For the filtered states $\mathbf{x}_t|y_{1:t}$, we have
\begin{align}
    \hat{\mathbf{x}}_{t|t}&=\sum_{k=1}^{M_{t|t}} \delta_{t|t}^{k}\hat{\mathbf{x}}_{t|t}^{k},\\
    \bm{\Sigma}_{t|t}&=\sum_{k=1}^{M_{t|t}} \delta_{t|t}^{k} \left[\bm{\Sigma}_{t|t}^{k} + (\hat{\mathbf{x}}_{t|t}^{k}-\hat{\mathbf{x}}_{t|t})(\hat{\mathbf{x}}_{t|t}^{k}-\hat{\mathbf{x}}_{t|t})^{\top}\right]. 
\end{align}
The smoothed states $\mathbf{x}_t|\mathbf{y}_{1:N}$ and $\mathbf{x}_{t+1},\mathbf{x}_t|\mathbf{y}_{1:N}$ can be computed similarly following standard expectation rules.

\vspace{-0.2cm}
\subsection{Relation with the Kalman filter and smoother}
Consider a linear time-invariant state-space system of the form \eqref{eqn:general_system}-\eqref{eqn:general_system2}, where the noises $\mathbf{w}_{t} \in \mathbb{R}^{n}$ and $v_{t} \in \mathbb{R}$ are Gaussian-distributed stochastic processes, with zero mean and covariance matrices $\mathbf{Q}$ and $R$, respectively. The standard Kalman filter (KF, \cite{kalman1960}) yields the filtered and predicted PDFs $p(\mathbf{x}_t|y_{1:t})=\mathcal{N}\left(\mathbf{x}_t;\hat{\mathbf{x}}_{t|t},\bm{\Sigma}_{t|t} \right)$ and $p(\mathbf{x}_{t+1}|y_{1:t})=\mathcal{N}\left(\mathbf{x}_{t+1};\hat{\mathbf{x}}_{t+1|t},\bm{\Sigma}_{t+1|t} \right)$, respectively, where $p(\mathbf{x}_{1})= \mathcal{N}\left( \mathbf{x}_{1}; \bm{\mu}_1,\mathbf{P}_1\right),$ and 
\begin{align}
    \mathbf{K}_{t}&= \bm{\Sigma}_{t|t-1} \mathbf{C}^\top (R+\mathbf{C}\bm{\Sigma}_{t|t-1}\mathbf{C}^\top)^{-1},\label{gainkalmanfilter}\\
    \hat{\mathbf{x}}_{t|t} &=\hat{\mathbf{x}}_{t|t-1} +\mathbf{K}_t (y_t - \mathbf{C} \hat{\mathbf{x}}_{t|t-1}-\mathbf{D}u_t),\label{xttkalmanfilter}\\
    \bm{\Sigma}_{t|t}&= (\mathbf{I}_{n}-\mathbf{K}_{t}\mathbf{C})\bm{\Sigma}_{t|t-1},\label{sigmatkalmanfilter}\\
    \hat{\mathbf{x}}_{t+1|t} &= \mathbf{A}\hat{\mathbf{x}}_{t|t} +\mathbf{B}\mathbf{u}_t,\label{xtmtkalmanfilter}\\
    \bm{\Sigma}_{t+1|t}&=\mathbf{Q}+\mathbf{A}\bm{\Sigma}_{t|t}\mathbf{A}^{\top},\label{sigmakalmanfilter}
\end{align}
initialized by 
$\hat{\mathbf{x}}_{1|0} = \bm{\mu}_1$ and $\bm{\Sigma}_{1|0} = \mathbf{P}_1$. On the other hand, the Rauch–Tung–Striebel smoother, also known as Kalman smoother (KS, \cite{rauch1965}) is as follows. The smoothing PDF for $t=N$ is the PDF $p(\mathbf{x}_N|y_{1:N})=\mathcal{N}\left(\mathbf{x}_N;\hat{\mathbf{x}}_{N|N},\bm{\Sigma}_{N|N} \right)$, which corresponds to the PDF of the last iteration of the Kalman filter. For $t=N-1,\dots,1$, the smoothing PDFs are given by $p(\mathbf{x}_{t}|y_{1:N})=\mathcal{N}\left(\mathbf{x}_{t};\hat{\mathbf{x}}_{t|N},\bm{\Sigma}_{t|N} \right)$, where 
    \begin{align}
    \label{kalmansmootherk}
        \mathbf{K}_{t|t+1} &=\bm{\Sigma}_{t|t}\mathbf{A}^{\top}(\mathbf{Q}+\mathbf{A}\bm{\Sigma}_{t|t}\mathbf{A}^{\top})^{-1}\\
        \hat{\mathbf{x}}_{t|N} &= \hat{\mathbf{x}}_{t|t} + \mathbf{K}_{t|t+1}\prt{\hat{\mathbf{x}}_{t+1|N}-\hat{\mathbf{x}}_{t+1|t}} \\
        \label{bothsides}
        \bm{\Sigma}_{t|N} &= \bm{\Sigma}_{t|t}+\mathbf{K}_{t|t+1}\prt{\bm{\Sigma}_{t+1|N}-\bm{\Sigma}_{t+1|t}}\mathbf{K}_{t|t+1}^{\top}, 
    \end{align}
    and the corresponding cross-covariance matrix
    \begin{align}
    \label{crosscovariancekalman}
        \mathbf{M}_{t|N}\hspace{-0.08cm} = \hspace{-0.07cm} \bm{\Sigma}_{t|t}\mathbf{K}_{t\hspace{-0.02cm}-\hspace{-0.02cm}1|t}^{\top}\hspace{-0.08cm}+\hspace{-0.07cm}\mathbf{K}_{t|t+1}\hspace{-0.04cm}\prt{\mathbf{M}_{t\hspace{-0.02cm}+\hspace{-0.02cm}1|N}\hspace{-0.08cm}-\hspace{-0.07cm}\mathbf{A}\bm{\Sigma}_{t|t}}\hspace{-0.03cm}\mathbf{K}_{t\hspace{-0.02cm}-\hspace{-0.02cm}1|t}^{\top}
    \end{align}
    is initialized by $\mathbf{M}_{N|N}=(\mathbf{I}_n-\mathbf{K}_N\mathbf{C})\mathbf{A}\bm{\Sigma}_{N-1|N-1}$.

    The cross-covariance matrix in \eqref{crosscovariancekalman} admits a similar expression to the one obtained for the QGSS in \eqref{jointsmoothingsigma}. To see this, note that post-multiplying \eqref{bothsides} by $\mathbf{K}_{t-1|t}^\top$ yields
    \begin{align}
        &\bm{\Sigma}_{t|N}\mathbf{K}_{t-1|t}^\top = \bm{\Sigma}_{t|t}\mathbf{K}_{t-1|t}^\top \notag \\&+\mathbf{K}_{t|t+1}\prt{\bm{\Sigma}_{t+1|N}\mathbf{K}_{t|t+1}^{\top}-\bm{\Sigma}_{t+1|t}\mathbf{K}_{t|t+1}^{\top}}\mathbf{K}_{t-1|t}^\top.
    \end{align}
    From \eqref{kalmansmootherk} and \eqref{sigmakalmanfilter} we find that $\bm{\Sigma}_{t+1|t}\mathbf{K}_{t|t+1}^{\top} = \mathbf{A}\bm{\Sigma}_{t|t}$, by which we conclude that both matrix sequences $\mathbf{M}_{t|N}$ and $\bm{\Sigma}_{t|N}\mathbf{K}_{t-1|t}^\top$ satisfy \eqref{crosscovariancekalman}. The $\bm{\Sigma}_{t|N}\mathbf{K}_{t-1|t}^\top$ sequence is initialized by
    \begin{align}
        \bm{\Sigma}_{N|N}\mathbf{K}_{N-1|N}^\top &= \bm{\Sigma}_{N|N} \bm{\Sigma}_{N|N-1}^{-1} \mathbf{A}\bm{\Sigma}_{N-1|N-1} \\
        &= (\mathbf{I}_n - \mathbf{K}_N \mathbf{C}) \mathbf{A}\bm{\Sigma}_{N-1|N-1},
    \end{align}
    where we have used \eqref{kalmansmootherk} and \eqref{sigmatkalmanfilter} to rewrite $\mathbf{K}_{N-1|N}$ and $\bm{\Sigma}_{N|N}$, respectively. Since this initialization point corresponds exactly with $\mathbf{M}_{N|N}$ previously defined, we conclude that both sequences are equal for all $t$, i.e., $\mathbf{M}_{t|N} = \bm{\Sigma}_{t|N}\mathbf{K}_{t-1|t}^\top$.
\begin{figure*}
	\centering
	\includegraphics[width=\linewidth]{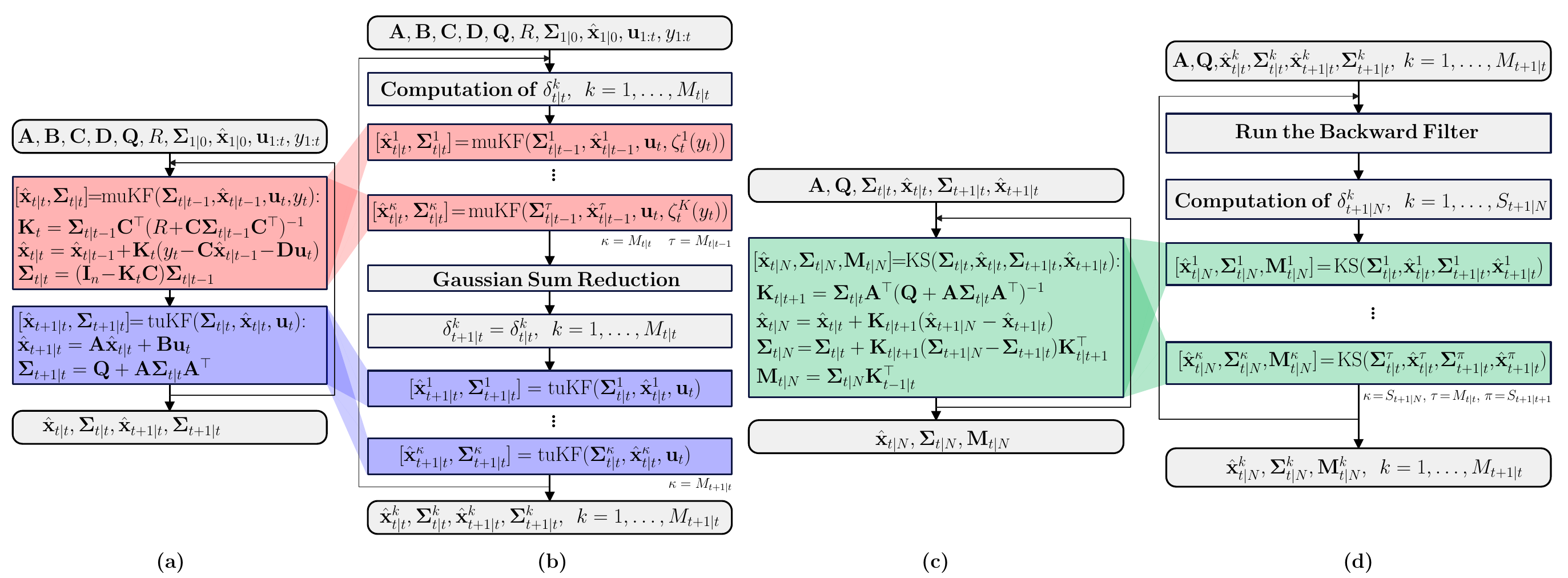}
	\vspace{-6mm}
	\caption{Parallel interpretation of the QGSF and QGSS. (a) pseudocode for Kalman measurement and time update stages, (b) block diagram of the QGSF indicating that during both the measurement and time update stages, several instances of the respective stages of the Kalman filter are employed concurrently, (c) pseudocode for the Kalman smoother, (d) block diagram of the QGSS, which uses the Kalman smoother equations in a parallelized manner.}
    	\vspace{-5mm}
	\label{comparative_with_kalman_filters_full}
\end{figure*}

Both the QGSF filter in Theorem \ref{thm:filtering} and the QGSS smoother given in Theorem \ref{thm:joint} contain iterations that can be related to the Kalman filter and smoother. The equations of the standard Kalman filter, in the measurement update stage, are found in \eqref{gainkalmanfilter}-\eqref{sigmatkalmanfilter}, which correspond to those found in \eqref{eqn:lemma_filtering_3}-\eqref{eqn:lemma_filtering_5} for the same stage in the QGSF filter, where the only difference is the superscript indicating the dependence on the Gaussian mixture model in \eqref{eqn:correction_filtering}. Similarly, the equations of the time update stage in \eqref{xtmtkalmanfilter}-\eqref{sigmakalmanfilter} correspond to the equations of the same stage, \eqref{eqn:lemma_filtering_7}-\eqref{eqn:lemma_filtering_8}, of the QGSF filter. In Figure \ref{comparative_with_kalman_filters_full}(a), the two stages of the Kalman filter are illustrated. Two functions \texttt{muKF()} and \texttt{tuKF()} have been defined in algorithmic form, which implement the corresponding equations of the measurement and time update stages, respectively. In Figure \ref{comparative_with_kalman_filters_full}(b), the same two stages of the QGSF filter are illustrated, and it is observed that each stage of this filter contains a defined number of instances of the \texttt{muKF()} and \texttt{tuKF()} functions, limited by $M_{t+1|t}=M_{t|t}$. This indicates that these instances can be solved in parallel, thereby reducing computation time. Furthermore, a similar behavior can be observed in the Kalman smoother given by the equations \eqref{kalmansmootherk}-\eqref{crosscovariancekalman} and their corresponding equations of the QGSS filter given in \eqref{jointsmoothingx}-\eqref{jointsmoothingsigma}. Figure \ref{comparative_with_kalman_filters_full}(c) illustrates the smoothing stage of the Kalman smoother, where once again the algorithmic function \texttt{KS()} has been defined, which implements the equations \eqref{kalmansmootherk}-\eqref{crosscovariancekalman}. In Figure \ref{comparative_with_kalman_filters_full}(d), it is observed that each stage of the QGSS algorithm requires $S_{t+1|N}$ instances of \texttt{KS()}, indicating that this algorithm can also be parallelized.

In summary, although the QGSF/QGSS are derived to solve the filtering and smoothing problem for general Wiener systems, the recursions that must be implemented are similar to the traditional equations of the KF/KS filter algorithms, and they can be interpreted as a set of Kalman filters and smoothers in parallel. This characteristic has the potential to be implemented in parallel hardware, which can result in a substantial reduction in computation time, making them an attractive alternative for embedded applications.

\vspace{-0.2cm}
\subsection{Relation with Multiple-model Kalman filters}
The Gaussian Sum filter and smoother derived in this work bear similarities to the multiple-model Kalman filter \cite{Hanlon2000,Khodarahmi2023}, where banks of Kalman filters run in parallel with different system parameter values, and their estimates are then combined using a hypothesis testing algorithm to obtain a final estimate of the system states. However, this method is not directly applicable to the problem addressed in this paper, particularly when there are segments of $g(\cdot)$ that evaluate to zero. For instance, consider a Wiener system with a rectifier function; in this scenario, the system can be viewed as a linear system with two possible models having different gains. However, one of these models becomes zero, resulting in the collapse of models into a single one, thus losing the multiple model representation.

\vspace{-0.2cm}
\subsection{Relation with Gaussian Sum - Quadrature Kalman filter}
The QKF and GS-QKF algorithms, derived in \cite{arasaratnam2007discrete}, can also be applied to state estimation of Wiener models, but have significant differences with respect to our proposed approach. The main differences with the proposed QGSF are as follows:
\begin{itemize}
    \item Contrary to this work, the QKF and GS-QKF algorithms only address the filtering problem (that is, not smoothing). 
    \item The QKF provides a Gaussian approximation of the posterior PDFs of the states, which, due to the system’s nonlinearities, are inherently non-Gaussian. In contrast, the proposed QGSF approximates the non-Gaussian measurement likelihood using a Gaussian mixture model, offering a non-Gaussian approximation that more accurately reflects the true likelihood.
    \item The QKF uses statistical linear regression to linearize the nonlinear system at each time step, yielding approximations of the nonlinear functions $\mathbf{x}_{t+1} = f(\mathbf{x}_{t}, \mathbf{u}_{t})$ and $\mathbf{z}_t = g(\mathbf{x}_t, \mathbf{u}_t)$ through hyperplanes that are linear in $\mathbf{x}_{t}$ and $\mathbf{u}_{t}$. This linearization introduces an error comparable to that of the EKF, as both methods rely on first-order approximations. In contrast, the proposed QGSF avoids approximating the nonlinear output function. Instead, it applies the random variable transformation theorem to derive the conditional probability function $p(y_t | \mathbf{x}_t)$, expressed through an integral equation. This equation covers three scenarios when the segments of $g(\cdot)$ are: i) constant; ii) strictly monotonic; or iii) a combination of both. The integral is solved numerically using the Gauss-Legendre quadrature rule, resulting in a representation of $p(y_t | \mathbf{x}_t)$ as a Gaussian mixture. By propagating this representation through Bayesian filtering equations, a Gaussian sum filter naturally emerges.   
    \item The GS-QKF models external non-Gaussianity arising from non-Gaussian noises, which are approximated via Gaussian mixtures. In contrast, the proposed QGSF models intrinsic non-Gaussianity in the system states (which GS-QKF assumes to be Gaussian due to the QKF framework) caused by output nonlinearities that render $p(y_t|\mathbf{x}_t)$ non-Gaussian. This implies that, although the GS-QKF and the QGSF share a similar structure (Gaussian sum), they are fundamentally distinct filters since they address non-Gaussianities from different sources.
\end{itemize}
\section{Simulation examples}\label{sec:simulations}
This section presents numerical examples to analyze the performance of our proposed filtering and smoothing algorithms. We compare our methods with established techniques such as the extended Kalman filter and smoother (EKF/EKS), the unscented Kalman filter and smoother (UKF/UKS), the quadrature Kalman filter and smoother (QKF/QKS), and the particle filter and smoother (PF/PS). Note that for the setup in \eqref{eqn:general_system}-\eqref{eqn:general_system4}, the Gaussian Sum-Quadrature Kalman Filter in \cite{arasaratnam2007discrete} simplifies to the QKF. While simulations provide access to noiseless system states for comparison with all state estimations, the true filtering and smoothing of probability density functions when state and output noise is incorporated remain unknown. To assess the resulting PDFs, we adopt the particle filter with $2\cdot 10^{4}$ particles as our baseline or ground truth (GT), utilizing $p(y_t|\mathbf{x}_t)$ calculated via equation \eqref{eqn:integral_lemma_pytxt} and the Monte Carlo integration method. The EKF/EKS, UKF/UKS, and QKF/QKS filters are implemented by rewriting the system in \eqref{eqn:general_system}-\eqref{eqn:general_system4} as an extended state-space system, with new state and input vectors $\tilde{\mathbf{x}}_{t}=[ \begin{matrix} \mathbf{x}_{t}^\top, & r_{t}\end{matrix}]^\top$ and $\tilde{\mathbf{u}}_t=[\begin{matrix} \mathbf{u}_t^\top, &\mathbf{u}_{t+1}^\top\end{matrix}]^\top$, where $\mathbf{u}_{N+1}=\mathbf{0}$, and
\begin{equation}
	\tilde{\mathbf{A}}\hspace{-0.05cm}=\hspace{-0.1cm}\left[\begin{matrix}\mathbf{A}&\hspace{-0.1cm}\mathbf{0}\\\mathbf{C}\mathbf{A}& \hspace{-0.1cm}0\end{matrix}\right]\hspace{-0.05cm}, \tilde{\mathbf{B}}\hspace{-0.05cm}=\hspace{-0.1cm}\left[\begin{matrix}\mathbf{B}&\hspace{-0.1cm}\mathbf{0}\\\mathbf{C}\mathbf{B}&\hspace{-0.1cm}\mathbf{D}\end{matrix}\right]\hspace{-0.05cm},\tilde{\mathbf{C}}\hspace{-0.05cm}=\hspace{-0.1cm}[\begin{matrix}\mathbf{0}, &\hspace{-0.1cm}1\end{matrix}], \tilde{\mathbf{D}}=\mathbf{0},\label{eqn:exteded_system_qoutput}
\end{equation}
which defines the following state-space model
\begin{align}
    \tilde{\mathbf{x}}_{t+1}&=\tilde{\mathbf{A}}\tilde{\mathbf{x}}_{t}+ \tilde{\mathbf{B}}\tilde{\mathbf{u}}_{t}+\tilde{\mathbf{w}}_{t},\\
    y_{t}&= g(\tilde{\mathbf{C}}\tilde{\mathbf{x}}_t)+\eta_t,
\end{align}
where $\tilde{\mathbf{w}}_t=[\begin{matrix} \mathbf{w}_t^\top, \mathbf{C}\mathbf{w}_t+v_{t+1}\end{matrix}]^\top$ is a Gaussian-distributed noise of zero mean and covariance matrix $\tilde{\mathbf{Q}}$ given by
\begin{equation}
	\tilde{\mathbf{Q}}=\left[\begin{matrix} \mathbf{Q}&\mathbf{Q}\mathbf{C}^\top \\ \mathbf{C}\mathbf{Q}^\top &\mathbf{C}\mathbf{Q}\mathbf{C}^\top +R\end{matrix}\right].
\end{equation}
A total of 100 Monte Carlo runs are performed, where we use 10 points of the Gauss-Legendre quadrature to approximate the PDF $p(y_t|\mathbf{x}_t)$ for the QGSF/QGSS implementation, i.e., $L=10$. We use $500$ particles for the implementation of the PF/PS algorithms, where PF was implemented using the systematic resampling technique \cite{Li2015b}, and $6$ sigma points are used for the QKF/QKS algorithms. We simulate $N=100$ input and output data. The experiments are carried out on a computer with the following specifications: Intel(R) Core(TM) i5-8300H CPU @ 2.30GHz 2.30 GHz processor, 8.00 GB RAM, Windows 11 operating system running MATLAB 2023b.

\subsection{\textbf{Example 1:} First-order system}
In the following example, we consider a first-order Wiener state-space system, where the linear block is described by $A=0.9$, $B=2.5$, $C=1.1$ and $D=1.5$, and with output
\begin{align}
    z_{t}&=r_{t}^2,\\
    y_{t}&= z_{t}+\eta_t,
\end{align}
where $w_{t}\sim \mathcal{N}\left(w_{t};0,1\right)$, $v_{t}\sim\mathcal{N}\left(v_{t};0,0.5\right)$, and $\eta_{t}\sim\mathcal{N}\left(\eta_{t};0,0.5\right)$. Additionally, we consider the input signal $u_t$ sampled from $\mathcal{N}\left(0,2\right)$, and $p(x_1)=\mathcal{N}(x_1;1,1)$. Notice that, for the quadratic function $z_t = r_t^2$ we have $\gamma_{1,2}(z_t) = \pm \sqrt{z_t}$ and $\phi_{1,2}(z_t) = 0.5/\sqrt{z_t}$. In Fig. \ref{pdf_filtering_1states_gsquare_state1}, we display the filtered PDFs for various time instances. 

Results obtained using the EKF, QKF and UKF algorithms diverge significantly from the ground truth PDFs. In sharp contrast, the proposed algorithm yields highly precise fits to the GT PDFs. Moreover, we observe that the PF with fewer particles produces less accurate fits to each GT PDF. Increasing the number of particles improves results, albeit at a high computational cost. Fig. \ref{cloud_filtering_1states_gsquare} illustrates the outcome of the Monte Carlo analysis with 100 experiments for estimating the state $\hat{x}_{t|t}=\mathbb{E}\left\lbrace x_t|y_{1:t}\right\rbrace$. The shaded region delineates the area encompassing all state estimations, while the red line represents the mean of the 100 estimates. The variability obtained using EKF, QKF and UKF is considerably higher compared to the variability obtained using QGSF and PF.
\begin{figure}
	\centering
	\includegraphics[width=\linewidth]{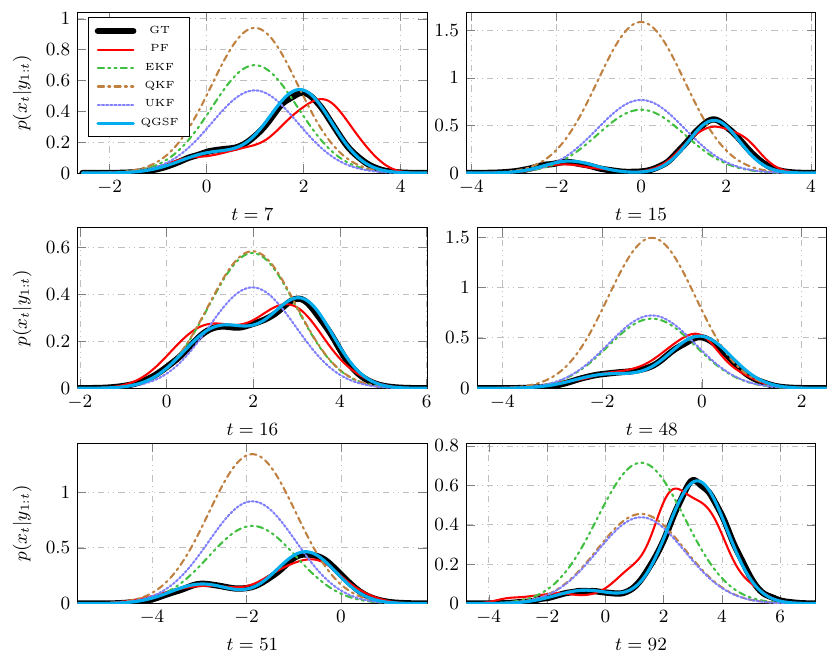}
	\vspace{-8mm}
	\caption{Filtered PDF \(p(x_t|y_{1:t})\) for selected time instances. We considered 10 Gaussian components for the QGSF, 6 sigma points for the QKF, 500 particles for the PF, and $\alpha=0.001$, $\beta=1$. $\kappa=0.001$ for the UKF.}
	\vspace{-3mm}
	\label{pdf_filtering_1states_gsquare_state1}
\end{figure}
\begin{figure}[!hb]
	\centering
	\includegraphics[width=1\linewidth]{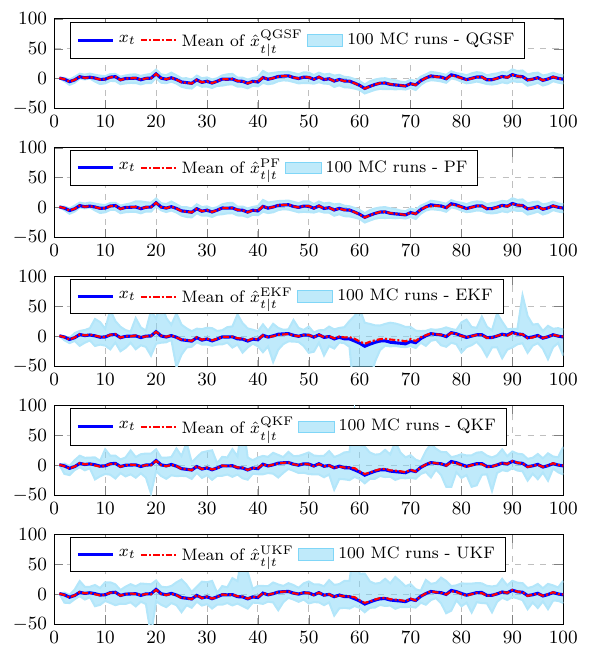}
	\vspace{-6mm}
	\caption{State estimation of the system $\hat{x}_{t|t}=\mathbb{E}\left\lbrace x_t|y_{1:t}\right\rbrace$ for 100 Monte Carlo runs. The shaded area depicts the region encompassing all state sequence estimates of the system. The blue line represents the true state, while the red line represents the mean of the 100 Monte Carlo experiments.}
	\label{cloud_filtering_1states_gsquare}
\end{figure}

In Fig. \ref{pdf_smoothing_1states_gsquare_state1}, we show the smoothed PDFs $p(x_t|y_{1:N})$ at different time instants. Results from EKS, QKS and UKS algorithms, just like their filtering counterparts, are quite different from the true PDFs. On the other hand, our proposed algorithm gives a close match to the true PDFs, followed closely by the PS algorithm. Fig. \ref{cloud_filtering_1states_gsquare} presents the results of Monte Carlo analysis with 100 experiments to estimate the state $\hat{x}_{t|N}=\mathbb{E}\left\lbrace x_t|y_{1:N}\right\rbrace$. The shaded area covers all state estimates, and the red line shows their average. Similar to the filtering example, the variability seen with the EKS, QKS and UKS estimates is much higher than with QGSS and PS.
\begin{figure}
	\centering
	\includegraphics[width=\linewidth]{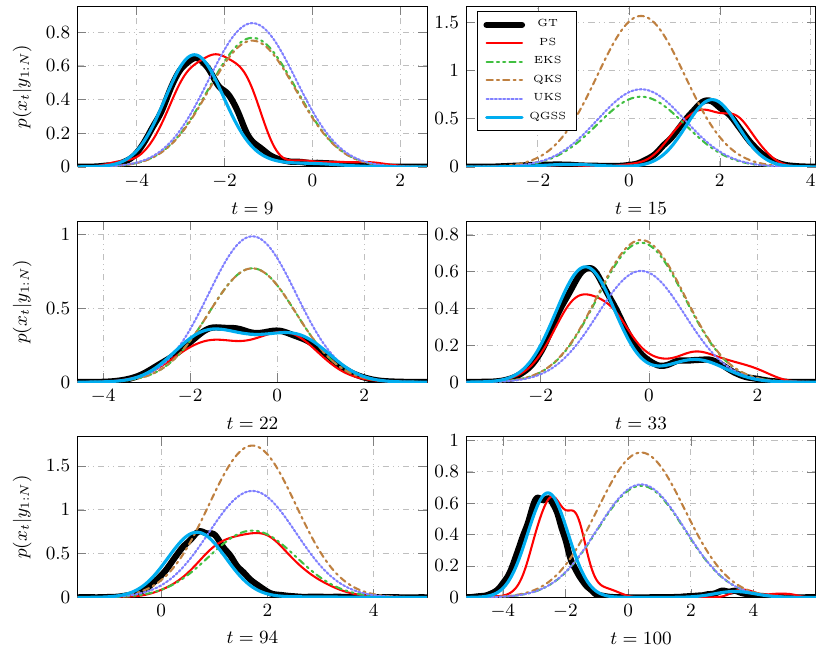}
	\vspace{-8mm}
	\caption{Smoothed PDF \(p(x_t|y_{1:N})\) for selected time instances. We considered 10 Gaussian components for the QGSS filter, 6 sigma points for the QKS filter, and 500 particles for the PS filter.}
    \label{pdf_smoothing_1states_gsquare_state1}
\end{figure}

\begin{figure}
	\centering
	\includegraphics[width=1\linewidth]{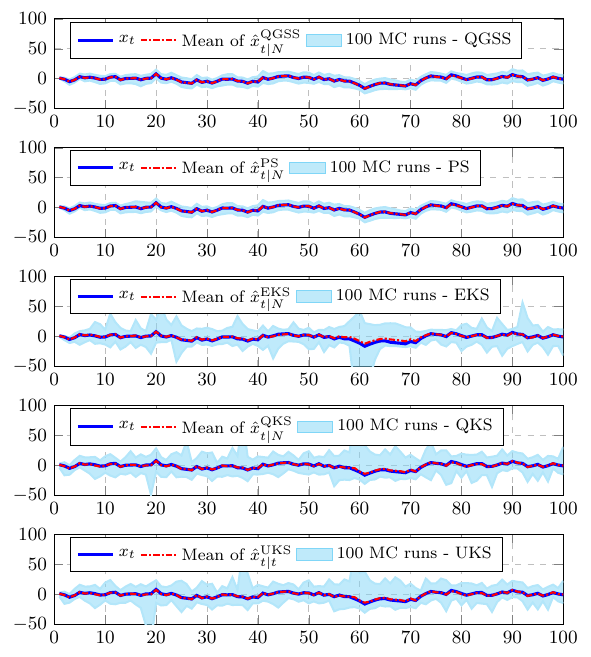}
	\vspace{-8mm}
	\caption{State estimation of the system $\hat{x}_{t|t}=\mathbb{E}\left\lbrace x_t|y_{1:N}\right\rbrace$ for 100 Monte Carlo        runs. The shaded area depicts the region encompassing all state sequence estimations of the system. The blue line represents the true state, while the red line represents the mean of the 100 Monte Carlo runs.}
	\label{cloud_smoothing_1states_gsquare}
\end{figure}

\subsection{\textbf{Example 2}: Second-order system}
In the following example, we explore a second-order Wiener state-space system where the parameter of \eqref{eqn:general_system}-\eqref{eqn:general_system2}:
\begin{align}
    \mathbf{A}\!=\!\left[\begin{matrix}0.9&0.1\\-0.1&0.7\end{matrix}\right]\!, \mathbf{B}\!=\!\left[ \begin{matrix}1.5\\2.5\end{matrix}\right]\!, \mathbf{C}\!=\!\left[ \begin{matrix}1.1&\!\!\!0.3\end{matrix}\right]\!, \mathbf{D}\!=\! 1.2,
\end{align}
and the output in \eqref{eqn:general_system3}-\eqref{eqn:general_system4} are defined as follows:
\begin{align}
    z_{t}&=\left\lbrace\begin{matrix} |r_{t}| & \textrm{if}& r_t\leq 0,\\ r_{t}^2 & \textrm{if}& r_t> 0,\end{matrix} \right.\\
    y_{t}&= z_{t}+\eta_t,
\end{align}
where $\mathbf{w}_{t}\sim \mathcal{N}\left(\mathbf{w}_{t};[0, 0]^\top,\mathbf{I}_2\right)$, $v_{t}\sim\mathcal{N}\left(v_{t};0,0.5\right)$, and $\eta_{t}\sim\mathcal{N}\left(\eta_{t};0,0.5\right)$. Additionally, we sample the input $u_t$ from $\mathcal{N}\left(0,2\right)$, and $p(\mathbf{x}_1)=\mathcal{N}(\mathbf{x}_1;[1, 1]^\top,\mathbf{I}_2)$. Notice that, $\gamma_{1}(z_t) = -z_t$, $\gamma_{2}(z_t) = \sqrt{z_t}$, $\phi_1(z_t) = 1$, and $\phi_2(z_t) =0.5/\sqrt{z_t}$. In Fig. \ref{data_error_filtering_2states_gpiece}, we illustrate box plots of the estimation error between the real state sequence and the sequence estimated by all the filtering and smoothing techniques considered in this comparison study. We observe that the estimation error of QGSF/QGSS is lower than PF/PS, EKF/EKS, and QKF/QKS, which exhibit, in some cases, high estimation errors.
\begin{figure}
	\centering
	\includegraphics[width=1\linewidth]{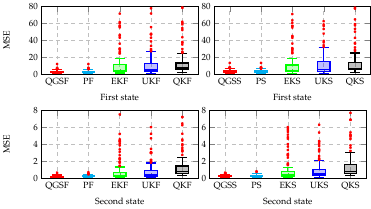}
	\vspace{-8mm}
	\caption{Estimation error between the real state sequence and the sequence estimated by QGSF/QGSS, PF/PS, EKF/EKS, and QKF/QKS.}
    	\vspace{-3mm}
	\label{data_error_filtering_2states_gpiece}
\end{figure}

\subsection{\textbf{Example 3}: Fourth-order, two-input system}
In this third case, we consider a fourth-order Wiener system in state-space, with two inputs and linear block described by

\begin{align}
    \mathbf{A}&\!=\!\left[\setlength\arraycolsep{2pt}\begin{matrix}0.52&0.4&0&0\\-0.4&0.52&0&0\\0&0&0.4&0.6\\0&0&0.06&-0.4\end{matrix}\right]\!, \quad \mathbf{B}\!=\!\left[ \setlength\arraycolsep{2pt}\begin{matrix}0.56&-0.58\\1.1&0.5\\5.3&-0.8\\-1.9&-0.45\end{matrix}\right]\!,\\
    \mathbf{C}&\!=\!\left[ \begin{matrix}0.5,&0.1,&0.5,&0.7\end{matrix}\right]\!,  \quad \mathbf{D}\!=\!\left[\begin{matrix}0,&0\end{matrix}\right]\!.
\end{align}
The outputs in \eqref{eqn:general_system3}-\eqref{eqn:general_system4} are defined as follows:
\begin{align}
    z_{t}&=\left\lbrace\begin{matrix}
		r_t+3 & \textrm{ if } & r_t<-3,\\
		0 & \textrm{ if } & -3\leq r_t< 3,\\
		r_t-3 & \textrm{ if } & r_t\geq 3,\\
	\end{matrix}\right.,\\
		y_{t}&= z_{t}+\eta_t,
\end{align}
where $\mathbf{w}_{t}\sim \mathcal{N}\left(\mathbf{w}_{t};[0,0,0,0]^\top,2\mathbf{I}_4\right)$, $v_{t}\sim \mathcal{N} \left(v_{t};0,1\right)$, and $\eta_{t}\sim\mathcal{N}\left(\eta_{t};0,0.5\right)$. We consider that the input signal is sampled from $\mathbf{u}_t \sim \mathcal{N}\left(\mathbf{u}_t;0,2\right)$, and the initial state is obtained from $p(\mathbf{x}_1)=\mathcal{N}(\mathbf{x}_1;[1,1,1,1]^\top,\mathbf{I}_4)$. The nonlinearity considered in this example corresponds to a deadzone function. Notice that, $\gamma_{1}(z_t) = z_t-3$, $\gamma_{2}(z_t) = z_t+3$, and $\phi_{1,2}(z_t) = 1$.

In Fig. \ref{data_error_filtering_4states_gdeadzone} we present boxplots of the mean square error of the states for each algorithm under study. Unlike the results obtained with the EKF/EKS and QKF/QKS approaches, very precise estimates are obtained using the proposed QGSF/QGSS algorithm, as well as the PF/PS algorithm. 

\begin{figure}
	\centering
	\includegraphics[width=1\linewidth]{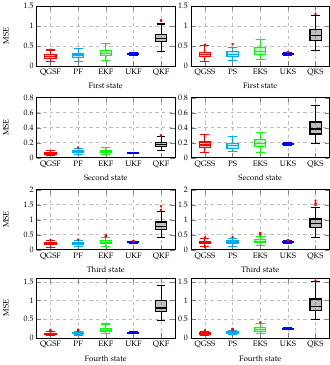}
	\vspace{-7mm}
	\caption{Example 3. Estimation error between the real state sequence and the sequence estimated by QGSF/QGSS, PF/PS, EKF/EKS, and QKF/QKS.}
	\label{data_error_filtering_4states_gdeadzone}
    	\vspace{-2mm}
\end{figure}

\subsection{Computational aspects}
We have also measured the computation time of all the algorithms under study, for each example. In Table \ref{tab:my-table1} we present the average time that one algorithm takes to execute the estimations for one Monte Carlo run. We observe that the proposed QGSF and QGSS algorithms are faster than PF/PS, while their performance is similar. Furthermore, contrary to the results of the QKF/QKS algorithms, their computation time does not grow significantly if the order of the system increases.

\begin{table}[]
	\renewcommand{\arraystretch}{1.3}
	\caption{Average time (seconds) of the Monte Carlo runs for Examples 1, 2, and 3.}
	\label{tab:my-table1}
	\centering
	\begin{tabular*}{\columnwidth}{@{\extracolsep{\fill}}|c|c|c|c|}    
		\hline
		\textbf{Algorithm} & \textbf{Example 1} & \textbf{Example 2} & \textbf{Example 3} \\ \hline
        QGSF/QGSS       & 0.2741/0.7230    & 0.3998/1.0577   & 0.7160/2.5557  \\ \hline
        PF/PS         & 0.7897/1.1257    & 2.7620/4.2627   & 3.7784/6.1774  \\ \hline
		EKF/EKS       & 0.0727/0.1424    & 0.0884/0.1686   & 0.1181/0.2195  \\ \hline
		QKF/QKS     & 0.0444/0.0475    & 0.1629/0.1660   & 4.3888/4.4683  \\ \hline
        UKF/UKS     & 0.0114/0.0024    & 0.0166/0.0027   & 0.0221/0.0027  \\ \hline
	\end{tabular*}
    	\vspace{-3mm}
\end{table}

\section{Conclusions}\label{sec:conclusions}
We presented a filtering and smoothing method for Wiener systems that accommodates a broad range of nonlinearities, including saturation, dead zones, rectification, and polynomial nonlinearities. To derive the filtering and smoothing algorithms, we use an explicit Gaussian quadrature approximation of the conditional PDF of the nonlinear output. This approach yields closed-form recursive formulas for both the filtering and smoothing distributions of the states, including the joint distribution at two consecutive instants. Extensive simulations show that the QGSF and QGSS algorithms produce more accurate state estimates than the extended Kalman filter, the quadrature Kalman filter, the unscented Kalman filter, and the particle filter with few particles. Compared to the particle filter and smoother, our approach requires a low number of Gaussian components in relation to the number of particles needed to produce a similar result. The proposed algorithms applied to the simulation study in Section VI reduce the computational time by a factor of 3 to 7 for the filtering step, and by a factor of 1.5 to 4 for the smoothing step. Lastly, the parallelizable nature of the proposed method makes it highly suitable for hardware implementation in embedded applications.

\bibliographystyle{plain} 
\bibliography{bibliography.bib}

\appendices

\section{Proof of Lemma \ref{lem:pz}}
\label{appendix:lem:pz}
\begin{proof}
    For any $z\in \mathbb{R}-\{z_1^*,\dots,z_{M_2}^*\}$, the theorem of transformation of random variables \cite[p. 130]{papoulis1989} gives
\begin{align}
\label{exp1}
    p(z)\hspace{-0.1cm} =\hspace{-0.1cm} \sum_{i=1}^{K_z} \left|\frac{\textnormal{d}\gamma_i(z)}{\textnormal{d}z}\right|\hspace{-0.03cm}\mathcal{N}(\gamma_i(z);\mu,\hspace{-0.02cm}R)\hspace{-0.1cm} =\hspace{-0.1cm} \sum_{i=1}^{M_1}\hspace{-0.07cm} \phi_i(z)\mathcal{N}(\tilde{\gamma}_i(z);\mu,R),
\end{align}
where we have used the definitions in \eqref{gammatildelemma}. On the other hand, when $z\in \{z_1^*,\dots,z_{M_2}^*\}$, $p(z)$ is a GPDF of the form
\begin{equation}
\label{exp2}
    p(z) = \sum_{j=1}^{M_2} \mathbb{P}(z=z_j^*) \delta(z-z_j^*).
\end{equation}
Since the Dirac delta functions only have support in $\{z_1^*,\dots,z_{M_2}^*\}$, we can add \eqref{exp1} and \eqref{exp2} to obtain 
\begin{equation}
    p(z) = \sum_{i=1}^{M_1} \phi_i(z)\mathcal{N}(\tilde{\gamma}_i(z);\mu,R) + \sum_{j=1}^{M_2} \mathbb{P}(z=z_j^*) \delta(z-z_j^*),
\end{equation}
which is valid for any $z\in \mathbb{R}$. The probability $\mathbb{P}(z=z_j^*)$ can be computed as the integral in \eqref{equationlemmaZ}, concluding the proof.
\end{proof}

\section{Proof of Theorem \ref{thm:pytxt}}
\label{appendix:thm:pytxt}
\begin{proof}
    The PDF $p(y_t|\mathbf{x}_t)$ can be obtained by marginalizing over $\eta_t$ following $p(y_t|\mathbf{x}_t) = \int_{\mathbb{R}} p(y_t,\eta_t|\mathbf{x}_t)\textnormal{d}\eta_t$. Thus, expanding the joint PDF $p(y_t,\eta_t|\mathbf{x}_t)$ leads to
 \begin{align}\label{equation_py}
     p(y_t|\mathbf{x}_t) &= \int_{\mathbb{R}} p(\eta_t)p(y_t|\eta_t,\mathbf{x}_t)\textnormal{d}\eta_t, \nonumber\\
     &=\int_{\mathbb{R}}p(\eta_t)p(z_t|\mathbf{x}_t)|_{z_t=y_t-\eta_t}\textnormal{d}\eta_t. 
 \end{align}
After applying Lemma \ref{lem:pz} with the output $z_t=g(r_t)$, where $r_t\sim \mathcal{N}(r_t;\mathbf{Cx}_t+\mathbf{D}\mathbf{u}_t,R)$, \eqref{equation_py} can be written as
    \begin{align}
     p(y_t|\mathbf{x}_t) \hspace{-0.06cm}&= \notag \\
     &\hspace{-1.1cm}\sum_{i=1}^{M_1} \hspace{-0.07cm}\int_{\mathbb{R}}\hspace{-0.07cm}\phi_i(y_t\hspace{-0.08cm}-\hspace{-0.05cm}\eta_t) \mathcal{N}(\tilde{\gamma}_i(y_t\hspace{-0.08cm}-\hspace{-0.05cm}\eta_t);\mathbf{Cx}_t\hspace{-0.06cm}+\hspace{-0.06cm}\mathbf{D}\mathbf{u}_t,\hspace{-0.02cm}R)\mathcal{N}(\eta_t;0,P)\textnormal{d}\eta_t \notag  \\
     &\hspace{-1.25cm}+ \hspace{-0.08cm}\sum_{j=1}^{M_2} \hspace{-0.08cm}\int_{\mathbb{R}}\hspace{-0.07cm}\delta(y\hspace{-0.07cm}-\hspace{-0.07cm}\eta_t\hspace{-0.08cm}-\hspace{-0.07cm}z_j^*)\mathcal{N}(\eta_t;\hspace{-0.03cm}0,\hspace{-0.03cm}P)\textnormal{d}\eta_t \hspace{-0.13cm}\int_{r_t\in g^{-1}(z_j^*)}\hspace{-1.2cm}\mathcal{N}\hspace{-0.02cm}(r_{\hspace{-0.02cm}t};\hspace{-0.03cm}\mathbf{Cx}_t\hspace{-0.06cm}+\hspace{-0.06cm}\mathbf{D}\mathbf{u}_t,\hspace{-0.04cm}R)\textnormal{d}r_{\hspace{-0.02cm}t}. \label{eqn:second_sum_proof}
 \end{align}
 Noticing that the integral over $\eta_t$ in \eqref{eqn:second_sum_proof} is directly obtained by evaluating $\mathcal{N}(\eta_t;0,P)$ in $\eta_t = y_t-z_j^*$, we obtain \eqref{eqn:integral_lemma_pytxt}.
\end{proof}

\section{Proof of Theorem \ref{thm:filtering}} \label{proof:filtering}
\begin{proof}
Consider the Bayesian filtering equations given in \eqref{eqn:bayesian_filtering_meas} and \eqref{eqn:bayesian_filtering_time}. In the measurement update equation, the normalization constant $p(y_t|y_{1:t-1})$ is given by
\begin{equation}
    p(y_t|y_{1:t-1})  = \int_{\mathbb{R}^n} p(y_t|\mathbf{x}_t)p(\mathbf{x}_t|y_{1:t-1})\textnormal{d}\mathbf{x}_{t}.
\end{equation}
Since both the numerator and denominator in \eqref{eqn:bayesian_filtering_meas} are computed using the product $F(\mathbf{x}_t,y_t)=p(y_t|\mathbf{x}_t)p(\mathbf{x}_t|y_{1:t-1})$, we focus on computing $F(\mathbf{x}_t,y_t)$. Due to the recursivity of the filtering equations, and noting that $p(\mathbf{x}_t|y_{1:t-1})$ is the one-step delayed version of $p(\mathbf{x}_{t+1}|y_{1:t})$ and is a Gaussian mixture, we can assume it has the following form:
\begin{equation}
    p(\mathbf{x}_{t}|y_{1:t-1})=\sum_{\ell=1}^{M_{t|t-1}} \delta_{t|t-1}^{k}\mathcal{N}(\mathbf{x}_{t};\hat{\mathbf{x}}_{t|t-1}^{\ell},\bm{\Sigma}_{t|t-1}^{\ell}),
\end{equation}
where the product $F(\mathbf{x}_t,y_t)$, considering the model $p(y_t|\mathbf{x}_t)$ given in \eqref{eqn:y_approximation_rew}, is given by
\begin{align}
    F(\mathbf{x}_t,y_t)=\sum_{\kappa=1}^{K}\!\!\sum_{\ell=1}^{M_{t|t-1}} \!\!\varphi_t^{\kappa} \delta_{t|t-1}^{\ell}&\mathcal{N}(\zeta_t^{\kappa}(y_t);\mathbf{C}\mathbf{x}_t \! + \! \mathbf{D}\mathbf{u}_t,R) \nonumber\\
    \times&\mathcal{N}(\mathbf{x}_{t};\hat{\mathbf{x}}_{t|t-1}^{\ell},\bm{\Sigma}_{t|t-1}^{\ell}).\!
\end{align}
Using Lemma \ref{lm:marginal_from_coditional} in Appendix \ref{section:technicallemmas} with $\mathcal{N}(\mathbf{x}_{t};\hat{\mathbf{x}}_{t|t-1}^{\ell},\bm{\Sigma}_{t|t-1}^{\ell})$ as the PDF $p(\mathbf{x})$ and $\mathcal{N}(\zeta_t^{\kappa}(y_t);\mathbf{C}\mathbf{x}_t+\mathbf{D}\mathbf{u}_t,R)$ as the PDF $p(\mathbf{y}|\mathbf{x})$, we obtain
\begin{align}
    &F(\mathbf{x}_t,y_t)= \notag \\
    &\sum_{\kappa=1}^{K}\!\!\sum_{\ell=1}^{M_{t|t-1}} \!\!\!\hspace{-0.02cm}\varphi_t^{\kappa} \delta_{t|t-1}^{\ell}\mathcal{N}\!\prt{\zeta_t^{\kappa}\hspace{-0.02cm}(y_t);\mathbf{C}\hat{\mathbf{x}}_{t|t\hspace{-0.01cm}-\hspace{-0.01cm}1}^{\ell}\!\!+\!\mathbf{D}\mathbf{u}_t,R\!+\!\mathbf{C}\bm{\Sigma}_{t|t\hspace{-0.01cm}-\hspace{-0.01cm}1}^{\ell}\mathbf{C}^{\top}\hspace{-0.03cm}} \nonumber\\
    \label{eqn:proof_thm_filtering}
	&\hspace{-0.05cm}\times\hspace{-0.08cm}\mathcal{N}\hspace{-0.04cm}\!\prt{\hspace{-0.04cm}\mathbf{x}_{t};\hspace{-0.03cm}\hat{\mathbf{x}}_{t\hspace{-0.005cm}|\hspace{-0.005cm}t\hspace{-0.02cm}-\hspace{-0.02cm}1}^{\ell}\!\hspace{-0.06cm}+\hspace{-0.01cm}\!\mathbf{K}_t^{\ell}\hspace{-0.02cm}(\hspace{-0.01cm}\zeta_t^{\kappa}\hspace{-0.04cm}(\hspace{-0.01cm}y_t\hspace{-0.01cm})\!\hspace{-0.03cm}-\hspace{-0.01cm}\!\mathbf{C}\hat{\mathbf{x}}_{t\hspace{-0.005cm}|\hspace{-0.005cm}t\hspace{-0.02cm}-\hspace{-0.02cm}1}^{\ell}\!\hspace{-0.07cm}-\!\hspace{-0.02cm}\mathbf{D}\mathbf{u}_t\hspace{-0.01cm})\hspace{-0.01cm},\hspace{-0.05cm}(\mathbf{I}_n\!\hspace{-0.08cm}-\!\hspace{-0.03cm}\mathbf{K}_t^{\ell}\hspace{-0.02cm}\mathbf{C})\bm{\Sigma}_{\hspace{-0.06cm}t\hspace{-0.005cm}|\hspace{-0.005cm}t\hspace{-0.02cm}-\hspace{-0.02cm}1}^{\ell}\hspace{-0.045cm}}\!\hspace{-0.03cm},
\end{align}
where $\mathbf{K}_t^{\ell}=\bm{\Sigma}_{t|t-1}^{\ell}\mathbf{C}^{\top}(R+\mathbf{C}\bm{\Sigma}_{t|t-1}^{\ell}\mathbf{C}^{\top})^{-1}$. Note that, once a new measurement $y_t$ is obtained, the term $\mathcal{N}(\zeta_t^{\kappa}(y_t);\mathbf{C}\hat{\mathbf{x}}_{t|t-1}^{\ell}+\mathbf{D}\mathbf{u}_t,R+\mathbf{C}\bm{\Sigma}_{t|t-1}^{\ell}\mathbf{C}^{\top})$ is just a numerical coefficient. Thus, the double summation in $F(\mathbf{x}_t,y_t)$ can be rewritten as a single summation, defining a new index $k=(\ell-1)K+\kappa$ such that
\begin{equation}\label{eqn:Fxtyt1}
	F(\mathbf{x}_t,y_t)=\sum_{k=1}^{M_{t|t}} \bar{\delta}_{t|t}^{k}\mathcal{N}\left( \mathbf{x}_{t};\hat{\mathbf{x}}_{t|t}^{k},\bm{\Sigma}_{t|t}^{k}\right) ,
\end{equation}
where $M_{t|t}$, $\bar{\delta}_{t|t}^{k}$, $\hat{\mathbf{x}}_{t|t}^{k}$, $\bm{\Sigma}_{t|t}^{k}$ are defined in \eqref{eqn:lemma_filtering_0}, \eqref{eqn:lemma_filtering_2}, \eqref{eqn:lemma_filtering_4}, and \eqref{eqn:lemma_filtering_5}, respectively. Then, the normalization constant is $p(y_t|y_{1:t-1}) =  \int_{\mathbb{R}^n} F(\mathbf{x}_t,y_t)\textnormal{d}\mathbf{x}_{t}=\sum_{\ell=1}^{M_{t|t}} \bar{\delta}_{t|t}^{\ell}$ which means that the correction stage filtering equation is 
\begin{equation}\label{eqn:Fxtyt2}
	p(\mathbf{x}_t|y_{1:t})=\sum_{k=1}^{M_{t|t}} \delta_{t|t}^{k}\mathcal{N}\left( \mathbf{x}_{t};\hat{\mathbf{x}}_{t|t}^{k},\bm{\Sigma}_{t|t}^{k}\right) ,\\	
\end{equation}
where $\delta_{t|t}^{k} = \bar{\delta}_{t|t}^{k}/\sum_{s=1}^{M_{t|t}}\bar{\delta}_{t|t}^{s}$ are normalized weights. On the other hand, the prediction stage equation is obtained by solving the integral in \eqref{eqn:bayesian_filtering_time}, that is
\begin{align}
    p(\mathbf{x}_{t+1}|y_{1:t})&=\int_{\mathbb{R}^n} p(\mathbf{x}_{t+1}|\mathbf{x}_t)p(\mathbf{x}_t|y_{1:t})\textnormal{d}\mathbf{x}_t, \nonumber\\
		&=\sum_{k=1}^{M_{t|t}} \delta_{t|t}^{k}\int_{\mathbb{R}^n}\mathcal{N}(\mathbf{x}_{t+1};\mathbf{A}\mathbf{x}_t+\mathbf{B}\mathbf{u}_t,\mathbf{Q})\nonumber\\ \label{samelines1}  
        &\hspace{18mm}\times\mathcal{N}(\mathbf{x}_{t};\hat{\mathbf{x}}_{t|t}^{k},\bm{\Sigma}_{t|t}^{k})\textnormal{d}\mathbf{x}_t.
\end{align}
Using Lemma \ref{lm:marginal_from_coditional} with $\mathcal{N}(\mathbf{x}_{t};\hat{\mathbf{x}}_{t|t}^{\ell},\bm{\Sigma}_{t|t}^{\ell})$ and $\mathcal{N}(\mathbf{x}_{t+1};\mathbf{A}\mathbf{x}_t+\mathbf{B}\mathbf{u}_t,\mathbf{Q})$ as the marginal and conditional PDFs respectively,
\begin{align}
	&p(\mathbf{x}_{t+1}|y_{1:t})\notag \\
 &\hspace{-0.04cm}=\hspace{-0.09cm}\sum_{k=1}^{M_{t|t}} \hspace{-0.08cm}\delta_{t|t}^{k}\mathcal{N}\hspace{-0.08cm}\prt{\hspace{-0.04cm}\mathbf{x}_{t\hspace{-0.02cm}+\hspace{-0.02cm}1};\hspace{-0.03cm}\mathbf{A}\hat{\mathbf{x}}_{t|t}^{k}\hspace{-0.09cm}+\hspace{-0.07cm}\mathbf{B}\mathbf{u}_t,\hspace{-0.04cm}\mathbf{Q}\hspace{-0.07cm}+\hspace{-0.08cm}\mathbf{A}\bm{\Sigma}_{\hspace{-0.05cm}t|t}^{k}\mathbf{A}^{\hspace{-0.08cm}\top}\hspace{-0.03cm}}\hspace{-0.13cm}\int_{\mathbb{R}^n}\hspace{-0.17cm}\mathcal{N}\hspace{-0.06cm}\left(\mathbf{x}_{t}; \cdot,\cdot\right)\hspace{-0.03cm}\textnormal{d}\mathbf{x}_t \nonumber \\
\label{samelines2}
    &\hspace{-0.04cm}=\hspace{-0.09cm}\sum_{k=1}^{M_{t+1|t}} \hspace{-0.05cm}\delta_{t+1|t}^{k} \mathcal{N}\left(\mathbf{x}_{t+1};\hat{\mathbf{x}}_{t+1|t}^{k},\bm{\Sigma}_{t+1|t}^{k}\right),
\end{align}
where $M_{t+1|t}$, $\delta_{t+1|t}^{k}$, $\hat{\mathbf{x}}_{t+1|t}^{k}$ and $\bm{\Sigma}_{t+1|t}^{k}$ are defined in the statement of the theorem, completing the proof.
\end{proof}

\section{Proof of Theorem \ref{thm:backward}} \label{proof:backward}
\begin{proof}
    Consider the backward filtering equations in \eqref{eqn:bayesian_backward_prediction} and \eqref{eqn:bayesian_backward_measurement}. The proof is carried out by induction in reverse time. First, it is verified that the recursion holds for $t=N-1$, then it is assumed that it holds for $t=s+1$, and finally, it is verified that it holds for $t=s$. Note that the recursion starts at $t=N$ with $p(y_{N:N}|\mathbf{x}_N)=p(y_{N}|\mathbf{x}_N)$ that is obtained from the output model, that is
    \begin{equation}
    \label{pynxn_appendix}
        p(y_{N}|\mathbf{x}_N) = \sum_{k=1}^{K}\varphi_N^{k}\mathcal{N}(\zeta_N^{k}(y_N); \mathbf{C}\mathbf{x}_N+\mathbf{D}\mathbf{u}_N,R),
    \end{equation} 
    where we directly observe the equivalences in the theorem. On the other hand, to verify that the recursion is fulfilled at the instant of time $t=N-1$, the backward prediction equation \eqref{eqn:bayesian_backward_prediction} is solved with $t=N-1$, that is $p(y_{N:N}|\mathbf{x}_{N-1})=\int_{\mathbb{R}^n} p(y_{N:N}|\mathbf{x}_{N})p(\mathbf{x}_{N}|\mathbf{x}_{N-1})\textnormal{d}\mathbf{x}_{N}$, where $p(y_{N:N}|\mathbf{x}_{N})$ is given by \eqref{pynxn_appendix} and the PDF $p(\mathbf{x}_{N}|\mathbf{x}_{N-1})=\mathcal{N}(\mathbf{x}_N;\mathbf{A}\mathbf{x}_{N-1}+\mathbf{Bu}_{N-1},\mathbf{Q})$. Applying Lemma \ref{lm:marginal_from_coditional} with $\mathbf{y} = \zeta_N^k(y_{N:N}), \bm{\mu}=\mathbf{D}\mathbf{u}_{N:N}$ and $\bm{\psi}=\mathbf{A}\mathbf{x}_{N-1}+\mathbf{Bu}_{N-1}$, we have 
    \begin{align}
        &p(y_{N:N}|\mathbf{x}_{N})p(\mathbf{x}_{N}|\mathbf{x}_{N-1})=\sum_{k=1}^{S_{N|N}}\epsilon_{N|N}^{k}\mathcal{N}(\mathbf{x}_N; \cdot,\cdot) \nonumber\\
        &\times \mathcal{N}\prt{\zeta_N^{k}(y_N); \mathbf{C}(\mathbf{A}\mathbf{x}_{N-1}\hspace{-0.05cm}+\hspace{-0.05cm}\mathbf{Bu}_{N-1})\hspace{-0.05cm}+\hspace{-0.05cm}\mathbf{D}\mathbf{u}_N,R\hspace{-0.05cm}+\hspace{-0.05cm}\mathbf{C}\mathbf{Q}\mathbf{C}^\top},
    \end{align}
    where the mean and covariance of the first Gaussian term in does not depend on $\mathbf{x}_N$. After integrating with respect to $\mathbf{x}_N$ and redefining the variables, we obtain \eqref{eqn:backward_prediction} with all the terms defined by \eqref{eqn:backward_prediction_S} and \eqref{eqn:backward_prediction_P} being evaluated at $t=N-1$. Now, for the correction stage in \eqref{eqn:bayesian_backward_measurement}, evaluating at $t=N-1$ produces
    \begin{equation}\label{eqn:measurement_N_1}
        p(y_{N-1:N}|\mathbf{x}_{N-1}) =p(y_{N-1}|\mathbf{x}_{N-1})p(y_{N:N}|\mathbf{x}_{N-1}).
    \end{equation}
    The expression for $p(y_{N-1}|\mathbf{x}_{N-1})$ is obtained from the output model \eqref{eqn:y_approximation_rew} evaluated at $t=N-1$. Therefore, by performing the product in \eqref{eqn:measurement_N_1} and defining $\ell=(k-1)K+\tau$, we obtain
    \begin{align}
            &p(y_{N-1:N}|\mathbf{x}_{N-1})\!=\!\!\sum_{\ell=1}^{S_{N-1|N-1}}\!\!\epsilon_{N-1|N-1}^{\ell} 
            \mathcal{N}\hspace{-0.05cm}\left(\begin{bmatrix} \zeta_{N\hspace{-0.02cm}-\hspace{-0.02cm}1}^\tau (y_{N-1}) \\ \zeta_{N}^k (y_{N}) \end{bmatrix}\hspace{-0.03cm}; \right.
            \notag   \\
            &\left.\begin{bmatrix} \mathbf{C} \\ \mathbf{CA}\end{bmatrix}\hspace{-0.05cm}\mathbf{x}_{N\hspace{-0.02cm}-\hspace{-0.02cm}1}\hspace{-0.05cm}+\hspace{-0.05cm}\begin{bmatrix} \mathbf{D} & \hspace{-0.05cm}\mathbf{0}^\top \\ \mathbf{CB} & \hspace{-0.05cm}\mathbf{D}\end{bmatrix}\hspace{-0.05cm}\begin{bmatrix} \mathbf{u}_{N\hspace{-0.02cm}-\hspace{-0.02cm}1} \\ \mathbf{u}_{N}\end{bmatrix}\hspace{-0.03cm}, 
            \begin{bmatrix} R & \hspace{-0.05cm}\mathbf{0}^\top \\ \mathbf{0} & \hspace{-0.05cm}R\hspace{-0.03cm}+\hspace{-0.03cm}\mathbf{CQC}^\top\end{bmatrix}\right) ,
    \end{align}
    which matches \eqref{eqn:backward_measureement} when the quantities in \eqref{eqn:backward_measurement_S}-\eqref{eqn:backward_measurement_H} are evaluated at $t=N-1$. By applying a procedure similar to that carried out for $t=N-1$, it can be verified that the results shown in the Theorem \ref{thm:backward} are also verified at the instant $t=s$, since the expressions obtained are the same as those given in \eqref{eqn:backward_measurement_S}-\eqref{eqn:backward_measurement_H} evaluated at $t=s$. Therefore, it can be concluded that Theorem \ref{thm:backward} holds for all $t$.
\end{proof}

\section{Proof of Theorem \ref{thm:smoothing}} \label{proof:smoothing}
\begin{proof}
	Consider the smoothing equation given in \eqref{eqn:bayesian_smoothing}. This PDF, at $t=N$, corresponds to the measurement update step of the Theorem \ref{thm:filtering} at instant $t=N$. That is the PDF $p(\mathbf{x}_N|\mathbf{y}_{1:N})$ is obtained from the last iteration of the correction step of the filtering stage. For $t=N-1,\dots,1$, the smoothing equation is $p(\mathbf{x}_t|\mathbf{y}_{1:N}) \propto p (\mathbf{x}_t|\mathbf{y}_{1:t-1})p(\mathbf{y}_{t:N}|\mathbf{x}_t)$. Therefore, by combining the prediction equation \eqref{eqn:prediccion_filtering} of the filtering algorithm with the equation \eqref{eqn:backward_measureement} of the reverse filtering algorithm, and later applying Lemma \ref{lm:marginal_from_coditional}, we obtain
\begin{align}
    &p(\mathbf{x}_t|\mathbf{y}_{1:N}) \propto \sum_{\tau=1}^{M_{t|t-1}}\sum_{\ell=1}^{S_{t|t}} \delta_{t|t-1}^{\tau}\epsilon_{t|t}^\ell \notag \\
    & \hspace{-0.1cm} \times \mathcal{N}\!\prt{\zeta_{t:N}^{\ell}\hspace{-0.02cm}(y_{t:N});\mathcal{O}_{t|t}\hat{\mathbf{x}}_{t|t\hspace{-0.01cm}-\hspace{-0.01cm}1}^{\tau}\!\!+\!\mathcal{H}_{t|t}\mathbf{u}_{t:N},\mathcal{P}_{t|t}\!+\!\mathcal{O}_{t|t}\bm{\Sigma}_{t|t\hspace{-0.01cm}-\hspace{-0.01cm}1}^{\tau}\mathcal{O}_{t|t}^{\top}\hspace{-0.03cm}} \nonumber\\
	&\hspace{-0.1cm}\times\hspace{-0.08cm}\mathcal{N}\hspace{-0.04cm}\!\prt{\hspace{-0.04cm}\mathbf{x}_{t};\hspace{-0.03cm}\hat{\mathbf{x}}_{t\hspace{-0.005cm}|\hspace{-0.005cm}t\hspace{-0.02cm}-\hspace{-0.02cm}1}^{\tau}\!\hspace{-0.065cm}+\hspace{-0.01cm}\!\mathbf{K}_{t\hspace{-0.005cm}|\hspace{-0.005cm}N}^{\tau}\hspace{-0.02cm}(\hspace{-0.01cm}\zeta_{t:N}^{\ell}\hspace{-0.02cm}(y_{t:N}\hspace{-0.015cm})\hspace{-0.07cm}-\hspace{-0.06cm}\mathcal{O}_{t\hspace{-0.005cm}|\hspace{-0.005cm}t}\hat{\mathbf{x}}_{t\hspace{-0.005cm}|\hspace{-0.005cm}t\hspace{-0.01cm}-\hspace{-0.01cm}1}^{\tau}\!\!\hspace{-0.03cm}-\hspace{-0.02cm}\!\mathcal{H}_{t\hspace{-0.005cm}|\hspace{-0.005cm}t}\mathbf{u}_{t:N}\hspace{-0.03cm}),\hspace{-0.03cm}\bm{\Sigma}_{\hspace{-0.03cm}t\hspace{-0.005cm}|\hspace{-0.005cm}N}^{\tau}\hspace{-0.05cm}}\!\hspace{-0.03cm},
\end{align}
where $\mathbf{K}_{t|N}^{\tau}=\bm{\Sigma}_{t|t-1}^{\tau}\mathcal{O}_{t|t}^{\top}(\mathcal{P}_{t|t}+\mathcal{O}_{t|t}\bm{\Sigma}_{t|t-1}^{\tau}\mathcal{O}_{t|t}^{\top})^{-1}$, and $\bm{\Sigma}_{t|N}^{\tau} = (\mathbf{I}_n-\mathbf{K}_{t|N}^{\tau}\mathcal{O}_{t|t})\bm{\Sigma}_{t|t-1}^{\tau}$.	A new index $k=(\ell-1)M_{t|t-1}+\tau$ is defined, which produces
	\begin{equation}
		p(\mathbf{x}_t|\mathbf{y}_{1:N}) \propto\sum_{k=1}^{S_{t|N}}\bar{\delta}_{t|N}^{k}\mathcal{N}(\mathbf{x}_t;\hat{\mathbf{x}}_{t|N}^{k},\bm{\Sigma}_{t|N}^{k}),
	\end{equation}
	where $S_{t|N}$, $\bar{\delta}_{t|N}^{k}$,  $\hat{\mathbf{x}}_{t|N}^{k}$, and $\bm{\Sigma}_{t|N}^{k}$ are defined in \eqref{eqsmoothingS} to \eqref{eqsmoothingsigma}. The weights are normalized by $\delta_{t|N}^{k}=\bar{\delta}_{t|N}^{k}/\sum_{s=1}^{S_{t|N}}\bar{\delta}_{t|N}^{s}$ following the approach in \eqref{eqn:Fxtyt1}-\eqref{eqn:Fxtyt2}, concluding the proof.
\end{proof}

\section{Proof of Theorem \ref{thm:joint}} \label{proof:joint}
\begin{proof}
	Consider \eqref{eqn:bayes_joint_smoothing} where the functions $p(\mathbf{x}_t|y_{1:t})$, and $p(y_{t+1:N}|\mathbf{x}_{t+1})$ are obtained from \eqref{eqn:correction_filtering} and \eqref{eqn:backward_measureement}, respectively; and $p(\mathbf{x}_{t+1}|\mathbf{x}_t)$ is given by the system model in \eqref{eqn:prob_model_pxtm1_xt}. Following the same lines as in \eqref{samelines1} and \eqref{samelines2}, the following product can be computed:
\begin{align}
    p&(\mathbf{x}_{t+1}|\mathbf{x}_t)p(\mathbf{x}_{t}|y_{1:t}) = \sum_{\tau=1}^{M_{t|t}} \delta_{t|t}^\tau \mathcal{N}\prt{\mathbf{x}_{t+1}; \hat{\mathbf{x}}_{t+1|t}^\tau, \bm{\Sigma}_{t+1|t}^\tau} \notag \\
    \label{gaussian1}
    &\times \hspace{-0.06cm} \mathcal{N} \hspace{-0.06cm}\prt{\hspace{-0.02cm}\mathbf{x}_{t}; \hat{\mathbf{x}}_{t|t}^\tau \hspace{-0.05cm}+ \hspace{-0.05cm}\mathbf{K}_{t|t+1}^\tau\hspace{-0.03cm}(\mathbf{x}_{t+1} \hspace{-0.05cm}- \hspace{-0.05cm}\hat{\mathbf{x}}_{t+1|t}^\tau), (\mathbf{I}_n \hspace{-0.06cm}- \hspace{-0.06cm}\mathbf{K}_{t|t+1}^\tau\mathbf{A})\bm{\Sigma}_{t|t}^\tau} \hspace{-0.07cm},
\end{align}
 where $\hat{\mathbf{x}}_{t+1|t}^\tau$ and $\bm{\Sigma}_{t+1|t}^\tau$ are as in \eqref{eqn:lemma_filtering_7} and \eqref{eqn:lemma_filtering_8}, respectively, and $\mathbf{K}_{t|t+1}^\tau=\bm{\Sigma}_{t|t}^\tau \mathbf{A}^\top (\mathbf{Q}+\mathbf{A}\bm{\Sigma}_{t|t}^\tau\mathbf{A}^\top)^{-1}$. The following identity is reached after applying Lemma \ref{lm:marginal_from_coditional}:
\begin{align}
    p&(y_{t+1:N}|\mathbf{x}_{t+1}) \mathcal{N}\hspace{-0.05cm}\prt{\hspace{-0.02cm}\mathbf{x}_{t+1}; \hat{\mathbf{x}}_{t+1|t}^\tau, \bm{\Sigma}_{\hspace{-0.02cm}t+1|t}^\tau} \hspace{-0.06cm} = \hspace{-0.15cm}\sum_{\ell=1}^{S_{t+1|t+1}}\hspace{-0.1cm}\epsilon_{t+1|t+1}^\ell \mathcal{C}_{t+1}^{\ell \tau} \notag \\
    \label{gaussian2}
    &\times \mathcal{N}\prt{\mathbf{x}_{t+1}; \hat{\mathbf{x}}_{t+1|N}^k, \bm{\Sigma}_{t+1|N}^k},
\end{align}
where we have used the notation in \eqref{eqsmoothingx} and \eqref{eqsmoothingsigma}, and $\mathcal{C}_{t+1}^{\ell \tau}$ is the Gaussian PDF defined in \eqref{eqsmoothingdelta} but evaluated at $t=t+1$. Thus, applying Lemma \ref{lem:joint_gaussian} with the Gaussian PDF in \eqref{gaussian2} as $p(\mathbf{x})$ and the second Gaussian in \eqref{gaussian1} as $p(\mathbf{y}|\mathbf{x})$ leads to
\begin{align}
    p(&y_{t+1:N}|\mathbf{x}_{t+1})p(\mathbf{x}_{t+1}|\mathbf{x}_t)p(\mathbf{x}_t|y_{1:t}) =  \notag \\
    &\sum_{\tau=1}^{M_{t|t}} \sum_{\ell=1}^{S_{t+1|t+1}}\hspace{-0.1cm}\delta_{t|t}^\tau\epsilon_{t+1|t+1}^\ell \mathcal{C}_{t+1}^{\ell \tau}\mathcal{N}(\mathbf{x}_{t}^{\text{e}};\hat{\mathbf{x}}_{t|N}^{\text{e}(k)},\bm{\Sigma}_{t|N}^{\text{e}(k)}),
\end{align}
where $\hat{\mathbf{x}}_{t|N}^{\text{e}(k)}$ and $\bm{\Sigma}_{t|N}^{\text{e}(k)}$ are defined in \eqref{jointsmoothingx} and \eqref{jointsmoothingsigma}, respectively. The normalization to obtain $p(\mathbf{x}_{t+1},\mathbf{x}_t|y_{1:N})$ leads to the weights being given by \eqref{eqsmoothingdelta}.

For $t=N$, Bayes' theorem leads to $p(\mathbf{x}_{N+1},\mathbf{x}_N|y_{1:N}) \propto p(\mathbf{x}_{N+1}|\mathbf{x}_N) p(\mathbf{x}_{N}|y_{1:N})$, where $p(\mathbf{x}_{N}|y_{1:N})$ is the last iteration of the filtering algorithm in the measurement stage. Thus,
\begin{align}
    &p(\mathbf{x}_{N+1},\mathbf{x}_{N}|y_{1:N}) \notag \\
    \label{appliedto}
    &\propto \hspace{-0.08cm}\sum_{k=1}^{M_{N|N}} \hspace{-0.08cm}\delta_{N|N}^{k}\mathcal{N}(\mathbf{x}_{N\hspace{-0.02cm}+\hspace{-0.02cm}1};\hspace{-0.02cm}\mathbf{A}\mathbf{x}_N\hspace{-0.08cm}+\hspace{-0.08cm}\mathbf{B}\mathbf{u}_N,\hspace{-0.02cm}\mathbf{Q})\mathcal{N}(\mathbf{x}_{N};\hat{\mathbf{x}}_{N|N}^{k},\hspace{-0.02cm}\bm{\Sigma}_{N|N}^{k}).
\end{align}
Lemma \ref{lem:joint_gaussian} applied above leads to the initialization conditions stated in the theorem, which concludes the proof.
\end{proof}

\section{Technical Lemmas}
\label{section:technicallemmas}

\begin{lemma}\label{lm:marginal_from_coditional}
Consider the PDF $p(\mathbf{x})=\mathcal{N}\left(\mathbf{x};\bm{\psi},\mathbf{Q}\right) $ and the conditional PDF  $p(\mathbf{y}|\mathbf{x})=\mathcal{N}\left(\mathbf{y}; \mathbf{C}\mathbf{x}+\bm{\mu},\mathbf{R}\right)$. Then, $p(\mathbf{x},\mathbf{y})$ can be decomposed as $p(\mathbf{x},\mathbf{y})=p(\mathbf{x}|\mathbf{y})p(\mathbf{y})$, where
\begin{align}
        p(\mathbf{x}|\mathbf{y})\hspace{-0.04cm}&=\hspace{-0.04cm}\mathcal{N}\left(\mathbf{x};\bm{\psi}\hspace{-0.04cm}+\hspace{-0.04cm}\mathbf{K}(\mathbf{y}\hspace{-0.04cm}-\hspace{-0.04cm}\mathbf{C}\bm{\psi}\hspace{-0.04cm}-\hspace{-0.04cm}\bm{\mu}),\mathbf{Q}\hspace{-0.04cm}-\hspace{-0.04cm}\mathbf{K}\mathbf{C}\mathbf{Q}\right), \\
p(\mathbf{y})&=\mathcal{N}\left(\mathbf{y};\mathbf{C}\bm{\psi}+\bm{\mu},\mathbf{R}+\mathbf{C}\mathbf{Q}\mathbf{C}^{\top}\right),
    \end{align}
and $\mathbf{K}=\mathbf{Q}\mathbf{C}^{\top}\left(\mathbf{R}+\mathbf{C}\mathbf{Q}\mathbf{C}^{\top}\right) ^{-1}$.
\end{lemma}	
\begin{proof}
    See Lemmas 2 and 3 of \cite[Section 6.6]{kitagawa1996smoothness}.
\end{proof}

\begin{lemma}\label{lem:joint_gaussian}
    Consider the PDF $p(\mathbf{x})=\mathcal{N}\prt{\mathbf{x};\bm{\psi},\mathbf{Q}}$ and the conditional PDF $p(\mathbf{y}|\mathbf{x})=\mathcal{N}\prt{\mathbf{y};\mathbf{C}\mathbf{x}+\bm{\mu},\mathbf{R}}$. Then, the joint PDF $p(\mathbf{x},\mathbf{y})$ is given by
    \begin{equation}
         p(\mathbf{x}\hspace{-0.02cm}, \hspace{-0.02cm}\mathbf{y})\hspace{-0.06cm}=\hspace{-0.04cm}\mathcal{N}\hspace{-0.07cm}\prt{\hspace{-0.04cm}\begin{bmatrix}
             \mathbf{x} \\ \mathbf{y}
         \end{bmatrix}\hspace{-0.06cm};\hspace{-0.04cm}\begin{bmatrix}
             \bm{\psi} \\ \mathbf{C}\bm{\psi} \hspace{-0.04cm}+\bm{\mu} \hspace{-0.04cm}
         \end{bmatrix}\hspace{-0.06cm},\hspace{-0.05cm}\begin{bmatrix}
             \mathbf{Q} & \hspace{-0.05cm}\mathbf{Q}\mathbf{C}^\top \\ \mathbf{C}\mathbf{Q} & \hspace{-0.05cm}\mathbf{R} \hspace{-0.05cm}+\hspace{-0.05cm} \mathbf{C}\mathbf{Q}\mathbf{C}^\top
         \end{bmatrix}\hspace{-0.02cm}}\hspace{-0.05cm}.
    \end{equation}
\end{lemma}
\begin{proof}
    The determinant of the covariance function associated with $p(\mathbf{x},\mathbf{y})$ above can be shown to be equal to $\det(\mathbf{Q})\det(\mathbf{R})$ by applying the Schur complement determinant formula \cite[Section 0.8.5]{horn2012}. Thus, what is left to verify is that the following equality holds:
    \begin{align}
        &(\mathbf{x}\hspace{-0.05cm}-\hspace{-0.05cm}\bm{\psi})^{\hspace{-0.04cm}\top} \mathbf{Q}^{-\hspace{-0.02cm}1}\hspace{-0.02cm}(\mathbf{x}\hspace{-0.05cm}-\hspace{-0.05cm}\bm{\psi})\hspace{-0.05cm}+\hspace{-0.05cm}(\mathbf{y}\hspace{-0.05cm}-\hspace{-0.05cm}\bm{\mu}\hspace{-0.05cm}-\hspace{-0.05cm}\mathbf{Cx})^{\hspace{-0.04cm}\top} \mathbf{R}^{-\hspace{-0.02cm}1}\hspace{-0.03cm}(\mathbf{y}\hspace{-0.05cm}-\hspace{-0.05cm}\bm{\mu}\hspace{-0.05cm}-\hspace{-0.05cm}\mathbf{Cx}) \notag \\
        &=\hspace{-0.1cm}\begin{bmatrix}
             \mathbf{x}-\bm{\psi} \\ \mathbf{y}\hspace{-0.09cm}-\hspace{-0.09cm}\bm{\mu}\hspace{-0.09cm}-\hspace{-0.09cm}\mathbf{C}\bm{\psi}
         \end{bmatrix}^{\hspace{-0.05cm}\top}\hspace{-0.1cm}\begin{bmatrix}
             \mathbf{Q} & \hspace{-0.08cm}\mathbf{Q}\mathbf{C}^\top \hspace{-0.05cm} \\ \mathbf{C}\mathbf{Q} & \hspace{-0.08cm}\mathbf{R} \hspace{-0.07cm}+\hspace{-0.06cm} \mathbf{C}\mathbf{Q}\hspace{-0.02cm}\mathbf{C}^{\hspace{-0.04cm}\top} \hspace{-0.1cm}
         \end{bmatrix}^{\hspace{-0.05cm}-\hspace{-0.03cm}1}\hspace{-0.1cm}\begin{bmatrix}
             \mathbf{x}-\bm{\psi} \\ \mathbf{y}\hspace{-0.09cm}-\hspace{-0.09cm}\bm{\mu}\hspace{-0.09cm}-\hspace{-0.09cm}\mathbf{C}\bm{\psi}
         \end{bmatrix} \hspace{-0.07cm}.
    \end{align}
    The verification of this equality is direct and follows standard algebraic manipulations.
\end{proof}

\begin{lemma}\label{lem:lemma_gmm_bif}
    Consider a conditional Gaussian PDF $p(\mathbf{y}|\mathbf{x})=\mathcal{N}\left(\mathbf{y};\mathcal{O}\mathbf{x}+\bm{\mu},\mathcal{P}\right)$ of the random variables $\mathbf{Y}\in\mathbb{R}^p$ and $\mathbf{X}\in\mathbb{R}^{n}$. Then, the backward filter form of this PDF is given by
    \begin{equation}\label{Normal1}
        p(\mathbf{y}|\mathbf{x})\!=\!\dfrac{1}{\sqrt{\det(2\pi \mathcal{P})}} \exp\left\lbrace -\frac{ 1}{2}\left(\mathbf{x}^\top\mathbf{F}\mathbf{x}-2\mathbf{G}^\top\mathbf{x}+H\right)\!\right\rbrace ,
    \end{equation}
    where $\mathbf{F}=\mathcal{O}^\top \mathcal{P}^{-1} \mathcal{O}$, $\mathbf{G}=\mathcal{O}^\top \mathcal{P}^{-1}(\mathbf{y}-\bm{\mu})$, and $H=(\mathbf{y}-\bm{\mu})^{\top}\mathcal{P}^{-1}(\mathbf{y}-\bm{\mu})$. If $\mathcal{O}$ has full column rank, then this representation can be rewritten as an unnormalized Gaussian distribution on the variable $\mathbf{X}$ as follows:
    \begin{equation}\label{Normal2}
        p(\mathbf{y}|\mathbf{x})=\alpha\mathcal{N}\left(\mathbf{x};\mathbf{F}^{-1} \mathbf{G},\mathbf{F}^{-1}\right),
    \end{equation}
    with $\alpha \!=\! \det\{2\pi \mathcal{P}\}^{\hspace{-0.03cm}-\frac{1}{2}}\hspace{-0.05cm}\det\{2\pi \mathbf{F}\}^{\hspace{-0.03cm}-\frac{1}{2}}\!\exp\!\left\lbrace\! -\frac{1}{2}(H \!\hspace{-0.02cm}-\! \mathbf{G}^\top \mathbf{F} ^{-1}\mathbf{G}) \hspace{-0.03cm}\right\rbrace$.
    \end{lemma}
    \begin{proof}
    Direct from expanding the exponential argument of $\mathcal{N}\left(\mathbf{y};\mathcal{O}\mathbf{x}+\bm{\mu},\mathcal{P}\right)$ and reordering the terms in the variable $\mathbf{x}$. On the other hand, Equation \eqref{Normal2} is obtained by completing the square in the exponential argument of \eqref{Normal1} and performing a normalization.
    \end{proof}
    
\end{document}